%% file: arxiv.tex
\documentclass[sigconf,nonacm]{acmart}

\usepackage{graphicx}
\usepackage{amsmath}
\usepackage{amsfonts}
\usepackage{amsthm}
\usepackage{algorithm}
\usepackage[noend]{algpseudocode}
\usepackage{textcomp}
\usepackage{xcolor}
\usepackage{times}
\usepackage{tikz}
\usepackage{paralist}
\usepackage{booktabs}
\usepackage{multirow}
\usepackage{subcaption}
\usepackage{url}
\usepackage{xspace}
\usepackage{balance}
\usepackage{relsize}
\usepackage[normalem]{ulem}
\usepackage{enumitem}
\usepackage{relsize}
\usepackage{balance}

\usetikzlibrary{shapes.geometric}
\usetikzlibrary{shapes,arrows}
\usetikzlibrary{positioning}
\usetikzlibrary{calc}

\tikzset{fontscale/.style = {font=\relsize{#1}}}

\newcommand{\figscale}{0.20}
\newcommand{\figcaption}[1]{\vspace*{-3mm}\caption{#1}\vspace*{-2mm}}

\newcommand{\tabcaption}[1]{\caption{#1}\vspace*{-3mm}}
\newcommand*\phantomrel[1]{\mathrel{\phantom{#1}}}
\renewcommand{\comment}[1]{}

\newtheorem{theorem}{Theorem}
\newtheorem{definition}{Definition}
\newtheorem{lemma}{Lemma}
\newtheorem{observation}{Observation}

\newcommand{\ab}[1]{{\color{red}{AB: #1}}}

\newcommand{\gaur}[1]{{\color{magenta}{GG: \it{#1}}}}

\usepackage[obeyFinal]{todonotes}

\newcommand{\yago}{YAGO2\xspace}
\newcommand{\gmark}{gMark\xspace}
\newcommand{\name}{\textsc{HaPPI}\xspace}

\newcommand{\provsql}{ProvSQL\xspace}
\newcommand{\huka}{HUKA\xspace}
\newcommand{\tripleprov}{TripleProv\xspace}

\newcommand{\pw}{\textsf{PossWorld}\xspace}
\newcommand{\ctd}{\textsf{PS-KC}\xspace}

\setlist{nosep,leftmargin=*}

\begin{document}

\title{Computing and Maintaining Provenance of Query Result Probabilities in Uncertain Knowledge Graphs}

\author{Garima Gaur}
\affiliation{%
  \institution{CSE, IIT Kanpur}
  \city{Kanpur }
  \country{India}}
\email{garimag@cse.iitk.ac.in}

\author{Abhishek Dang}
\affiliation{%
  \institution{CSE, IIT Kanpur}
  \city{Kanpur }
  \country{India}}
\email{ahdang@cse.iitk.ac.in}

\author{Arnab Bhattacharya}
\affiliation{%
  \institution{CSE, IIT Kanpur}
  \city{Kanpur }
  \country{India}}
\email{arnabb@cse.iitk.ac.in}

\author{Srikanta Bedathur}
\affiliation{%
  \institution{CSE, IIT Delhi}
  \city{New Delhi}
  \country{India}}
\email{srikanta@cse.iitd.ac.in}

\begin{abstract}
	Knowledge graphs (KG) that model the relationships between entities as
	labeled edges (or facts) in a graph are mostly constructed using a suite of
	automated extractors, thereby inherently leading to uncertainty in the
	extracted facts.  Modeling the uncertainty as probabilistic confidence
	scores results in a \emph{probabilistic knowledge graph}. Graph queries
	over such probabilistic KGs require answer computation along with the
	computation of those result probabilities, aka, probabilistic inference.
	We propose a system, \textsf{\name} (\textit{H}ow \textit{P}rovenance of
	\textit{P}robabilistic \textit{I}nference), to handle such query
	processing.  Complying with the standard provenance semiring model, we
	propose a novel \emph{commutative semiring} to \emph{symbolically compute}
	the probability of the result of a query. These provenance-polynomial-like
	\emph{symbolic expressions} encode fine-grained information about the
	probability computation process. We leverage this encoding to efficiently
	\emph{compute} as well as \emph{maintain} the probability of results as the
	underlying KG changes.  Focusing on a popular class of conjunctive basic
	graph pattern queries on the KG, we compare the performance of \name
	against a possible-world model of computation and a knowledge compilation
	tool over two large datasets. We also
	propose an \emph{adaptive} system that leverages the strengths of both \name
	and compilation based techniques.  Since existing systems for probabilistic
	databases mostly focus on query computation, they default to
	re-computation when facts in the KG are updated.  \name, on the other hand,
	does not just perform probabilistic inference and maintain their
	provenance, but also provides a mechanism to \emph{incrementally} maintain
	them as the KG changes.  We extend this maintainability as part of our
	proposed \emph{adaptive} system.
\end{abstract}


\maketitle

\input{introduction}
\input{background}
\input{related-work}

\input{solution}

\input{algo}
\input{analysis}

\input{experiment}

\input{closing}

\balance
\bibliographystyle{ACM-Reference-Format}
\bibliography{papers}
\clearpage

\appendix
\input{appendix}

\end{document}

%% file: introduction.tex
\section{Introduction}
\label{sec:intro}

Knowledge graphs (KGs) are central to many real-life systems such as search
engines, social networks, medical assistants, question answering, etc. Some of
these KGs are automatically built by employing a suite of knowledge extractors
and integrators.  Differences in various extraction approaches, inherent
ambiguities of the extraction process itself, and variations in the credibility
of data sources make the automatically extracted facts uncertain. The
uncertainty is encoded by assigning a \emph{confidence score} to each fact,
leading to \emph{probabilistic knowledge graphs} such as YAGO2~\cite{yago2},
NELL~\cite{NELL-aaai15}, ReVerb~\cite{reverb}, Probase~\cite{probase}, etc.

In addition, despite the phenomenal progress in information extraction
techniques, erroneous facts invariably creep into KGs. This, in turn, results
in query answers being erroneous.  For instance, the query for ``a list of
African comedians'' over the NELL KG includes ``Jimmy Fallon'' and ``Ellen
DeGeneres'' (as per NELL extraction round \#1082). Although both are comedians,
neither of them are from Africa. To determine the source of the error, or to
debug the KG, it is necessary to compute the \emph{provenance} of each result.

\begin{figure}[t]
\centering
  \centering
  \includegraphics[scale=0.18]{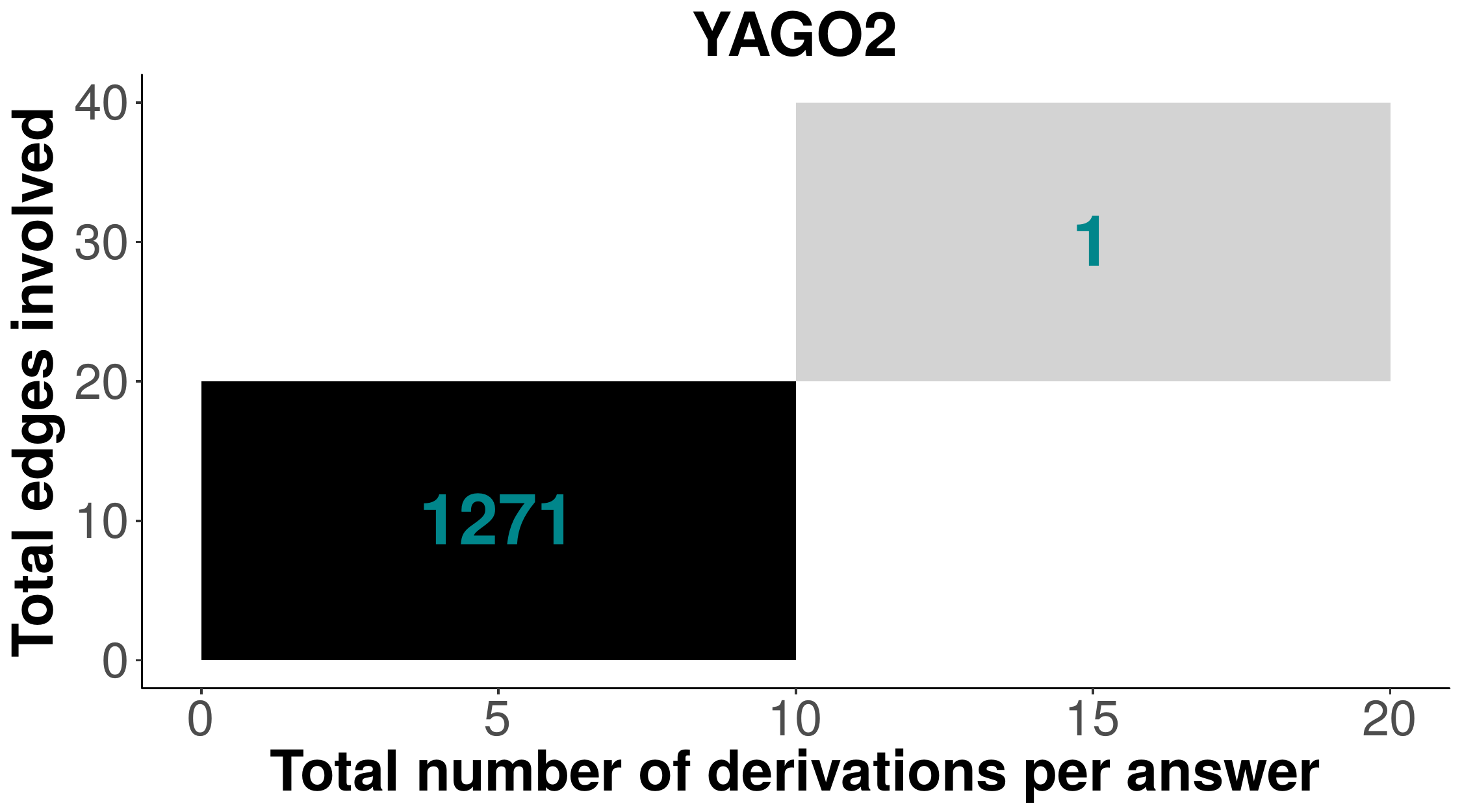}
	\hfill{}
  \includegraphics[scale=0.18]{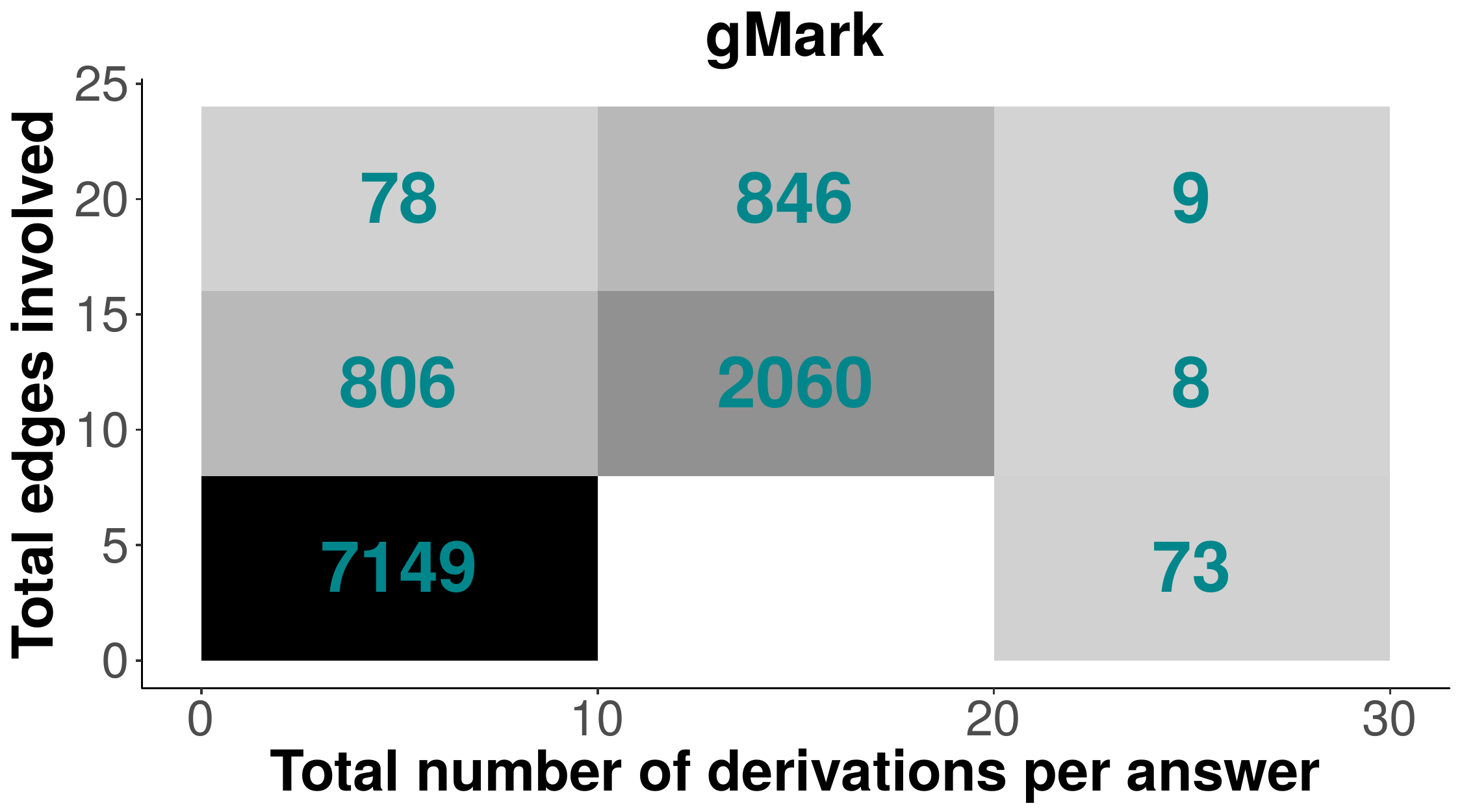}
	\figcaption{Distribution of \yago and \gmark query answers according to their total derivations and involved edges}
\label{fig:heat}
\end{figure}

Provenance represents the derivation process of an answer.  For probabilistic
data, typically, it is modeled as a Boolean formula in disjunctive normal form
(DNF), where each conjunct encodes a derivation.  This Boolean formula,
referred to as \emph{lineage}, is also used for probability computation of the
answer by counting all satisfying assignments -- equivalent to the model
counting problem which is known to be \#P-hard ~\cite{roth96}. The brute-force
way of counting satisfying assignments, known as possible world computation, is
to iterate over all possible assignments taking time exponential in the number
of the Boolean variables involved.  For provenance in probabilistic graphs, a
Boolean variable is associated with each edge.
A more nuanced approach to probability computation is the \emph{knowledge
compilation} technique that translates the lineage formula to a more tractable
Boolean circuit using SAT solvers.  Although it cannot guarantee scalability,
for large answers the use of compilation tools (e.g., C2D~\cite{c2d},
D4~\cite{d4}, dSharp~\cite{dsharp}) is known to be practical.
  
However, for answers with a small number of Boolean variables in their
derivation formula, the overhead of a SAT solver invocation results in a
considerably poor performance. Na\"ively using knowledge compilation tools may,
thus, fail to take advantage of small size of the computation problem. We
investigated the results of query workloads from \gmark~\cite{gmark} as well as
\yago~\cite{yago2} to understand the extent to which this behavior can affect
the overall performance. Fig.~\ref{fig:heat} summarizes the distribution of
answers with different number of derivations (along x-axis) and edge counts
(along y-axis), suitably bucketed for readability. We use color to indicate the
absolute count of answers in each region (darker the color, more the count),
and also print the raw count. It can be seen that most graph query answers are
concentrated in the area where the derivation and edge counts are
low---precisely where knowledge compilation is not the best method. In fact, in
this region, they are out-performed by even a possible-world computation that
employs brute-force evaluation.

In this paper, we specifically target this region where a large number of
answers are found, and present an algorithm for probability computations that
significantly speeds up the performance there. We implement this in a system
named \textbf{\name} (\underline{H}ow \underline{P}rovenance of
\underline{P}robabilistic \underline{I}nference). 
Thematically to the provenance semiring model~\cite{provenancesemirings}, we
also introduce a novel semiring which enables \name to symbolically compute the
answer probability. The proposed semiring facilitates computing \emph{how}
provenance of not just the query answer but also of its probability
computation. Unlike a Boolean formula that simply represents how edges interact
to generate an answer, we additionally capture the arithmetic involved to
compute the exact probability. This fine-grained provenance information allows
for efficient maintenance in \name.

\name outperforms knowledge compilation tools as well as possible world
computation for answers of the kind found in the highly populated bottom left
region of Fig.~\ref{fig:heat}.  Since knowledge compilation techniques work
best for answers with high derivation and edge counts, we also propose an
\emph{adaptive system} that uses \name for small answers and compilation
techniques for larger answers. We show that this adaptive system produces
sizeable gains in performance over either system used in isolation. Further,
since the adaptive system uses \name for a large number of answers, it inherits
the maintainability of \name.

The key \emph{contributions} of this paper are four-fold:  

\begin{enumerate}
     
     \item A novel theoretical model (Sec.~\ref{sec:framework}) based on a semiring to support efficient probability computation over probabilistic KGs.

	\item A practical implementation (Sec.~\ref{sec:algo}),
		\name\footnote{\url{https://github.com/gaurgarima/HaPPI}}, extends a provenance-aware property graph system,
\huka\footnote{\url{https://github.com/gaurgarima/HUKA}}.  Our algorithm can be also used in conjunction with other works, like \huka \cite{huka,gaur2017} and \provsql \cite{provsql},
based on the same underlying system to expand their support for probabilistic data.

		\item A theoretical analysis as well as an empirical evaluation
			highlighting the easy maintainability of \name under insertion of
			edges (i.e., new facts) to the KG.
			
		\item Finally, a proposal for an \emph{adaptive} framework that
			leverages the superior performance of \name and knowledge
			compilation tools at different regions of answer sizes. Our
			extensive empirical evaluation, using queries over \gmark,
a synthetic KG, and, \yago, a real-world KG, shows that this adaptive system
			outperforms any existing method used in isolation. 

\end{enumerate}

%% file: background.tex
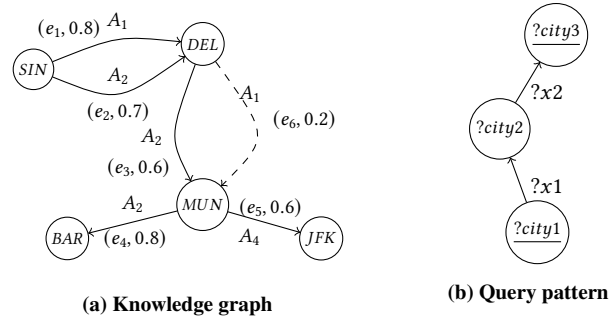
\begin{figure}[t]
	\begin{minipage}{0.6\columnwidth}
	\begin{subfigure}[t]{1\textwidth}
		\centering
		\begin{tikzpicture}[scale=0.45, transform shape]
			\tikzstyle{every node} = [circle,draw=black]
			\node (c1) at (-5,5) {\huge $SIN$};
			\node (c2) at (0,5.75) {\huge $DEL$};
			\node (c3) at (0,1) {\huge $MUN$};
			\node (c4) at (-4,0) {\huge $BAR$};
			\node (c5) at (3.5,0) {\huge $JFK$};
			\tikzstyle{every node} = [circle]
			
			\draw[->]  (c1)  .. controls(-2.5,6) .. (c2) node [midway,above,fill=none,fontscale=4] {$A_{1}$};     

			\draw[->]  (c1) .. controls(-2.5,4) .. (c2) node [midway,above,fill=none,fontscale=4] {$A_{2}$};

			\draw[->]  (c2)  .. controls(-1,3) .. (c3);     
			\node[fontscale=4] (a1) at (-1.5,3) { $A_{2}$};
			
		\draw[->]  (c3)  -- (c4) node [midway,above,fill=none,fontscale=4] {$A_{2}$};   
		
			\draw[->]  (c3) -- (c5) node [midway,below,fill=none,pos=0.3,fontscale=4] {$A_{4}$};     
		
\draw[->, dashed]  (c2)  .. controls(2,3) .. (c3) node [midway,above,pos=0.3,fill=none,fontscale=4] {$A_{1}$};

			\node[fontscale=3.5] (e1) at (-3.95,6.2) { $(e_{1},0.8)$};
			\node[fontscale=3.5] (e2) at (-2.5,3.75) { $(e_{2},0.7)$};			
				\node[fontscale=3.5] (e3) at (-1.9,2.1) { $(e_{3},0.6)$}; 
				
			\node[fontscale=3.5] (e4) at (-2,0) { $(e_{4},0.8)$};
			\node[fontscale=3.5] (e5) at (2,0.9) { $(e_{5},0.6)$};
			\node[fontscale=3.5] (e6) at (3,3.5) { $(e_{6},0.2)$};

\end{tikzpicture}   
\caption{Knowledge graph}
\label{fig:kg}
	\end{subfigure}
\end{minipage}\hfill
\begin{minipage}{0.3\columnwidth}
	\begin{subfigure}[t]{1\textwidth}
	\centering
		\begin{tikzpicture}[scale=0.5, transform shape]
			\tikzstyle{every node} = [circle, draw=black]
			
			\node(a) at (3,-1.75) {\huge \underline{$?city1$}};
			\node (b) at (2,1) {\huge $?city2$};
			\node (c) at (3.5,3.5) {\huge \underline{$?city3$}};
	
			\tikzstyle{every node} = [circle]
			
			\draw[->] (a) -- (b);
			\node[fontscale=3.5] (x1) at (3.25,-0.5) { $?x1$};
	
			\draw[->] (b) --(c);
			\node[fontscale=3.5] (x2) at (3.25,2) { $?x2$};			
	\end{tikzpicture}   
	\caption{Query pattern}
	\label{fig:query}
	\end{subfigure}
\end{minipage}

\figcaption{Running example of flight knowledge graph, a query for pairs of cities with one-stop connection between them, and a listing of edge probabilities.}
\vspace*{2mm}
\label{fig:runningEg}
\end{figure}

\section{Background}
\label{sec:background}

\subsection{Probabilistic Knowledge Graph}

A probabilistic knowledge graph is a graph $G \left(V, E, L, Id, Pr\right)$ with vertex-set $V$ representing the entities, labeled edge-set $E$ with each edge $e$ represented as $\langle u,l,v \rangle$ with $l \in L$ encoding the relation between two vertices $u$ and $v$. 
It is also common to refer to an edge $e = \langle u, l, v \rangle$ in the knowledge graph as a \emph{fact} where $u$, $l$, $v$ are subject, predicate and object of the fact respectively.
Each edge is assigned a unique id, ${Id}: E \to \mathcal{N}$. Further, we associate with each edge a value between $0$ and $1$, $Pr: E \to [0,1]$, representing the probability of the corresponding fact. We make the standard \emph{edge independence assumption} where the existence of an edge is independent of the other edges in the KG.
Unlike a deterministic KG, the presence of each fact in 
the KG is a probabilistic event, the probability of which is referred to as the \emph{existential probability} of the fact.

\subsubsection*{Possible World Semantics}

Equivalently, we can interpret the probabilistic knowledge graph as a 
collection of edge-induced subgraphs, called \emph{possible worlds},
$\mathcal{G_{PW}} = \{G_{1}, \dots, G_{n}\}$ of the knowledge graph $G$.
There is a probability distribution, $\mathcal{P}$, defined over all possible worlds, such that $\Sigma_{G_i \in G_{PW}} \mathcal{P}(G_i) = 1.$ With edge-independence assumption as before, $P(G_i)$ in terms of probability of its edges is
\begin{align}
\mathcal{P}(G_{i}) = \prod_{e_{i}\in E_{i}} Pr(e_{i}) \times \prod_{e_{i}\in E\setminus  E_{i}} ( 1 - Pr(e_{i}))
\end{align}
where $E_{i}$ is the set of edges present in $G_{i}$.

\subsection{Graph Query}
A graph query is formulated as a graph pattern that a user intends to find in
the knowledge graph $G$. Similar to the triple representation of KG, a graph
query expresses the query graph pattern as a collection of triples.  Each edge
of the graph pattern corresponds to a triple pattern in the query. Similar to
graph triples, a \emph{triple pattern} consists of subject, predicate, and
object.  Both subject and object can be variables or be bound to one of the
vertices of the KG. A predicate could be a variable or one of the labels of the
edges of the KG. A graph query pattern can be realized as a query using the SPARQL query language.

A graph query can be interpreted as a \emph{conjunction} of triple patterns and
the aim is to find all possible bindings to the variables of the triple
patterns as a whole. These conjunctive graph queries are popularly known as
\emph{Basic Graph Pattern} (BGP) queries~\cite{sparqlnew}.  Such SPARQL queries
can be expressed as relational SPJ (Select-Project-Join) queries
\cite{cyganiak2005}. 
In this work we are handling a subset of SPARQL queries which does not include more sophisticated operators such as Union, Optional, etc. We see inclusion of these SPARQL queries as future work. 

\subsection{Running Example} Throughout the paper we use the following running example. Consider an air-travel agency that provides a flight search engine. The search engine uses the knowledge graph shown in Fig.~\ref{fig:kg} as its knowledge base. The nodes in the graph denote the airport codes of different cities it operates in: Singapore (SIN), New Delhi (DEL), Munich (MUN), Barcelona (BAR), and New York (JFK). An edge between two cities represents a direct flight with edge label representing the airline operating the flight.
 The operation of each flight is dependent on different factors, like environmental, financial, political, etc. Thus, the existence of the edges in the graph is a probabilistic event. Each edge of the KG, shown in Fig.~\ref{fig:kg}, is annotated with its existential probability.
Suppose a user wants to list down all the pairs of cities that have  \emph{one-stop} connecting flights between them with a joint probability greater than $0.5$. The corresponding query pattern, shown in Fig.~\ref{fig:query}, and the equivalent SPARQL query, 

\texttt{
	\begingroup
		\fontsize{8pt}{12pt}
		\textbf{Select} \textit{?city1} \textit{?city2} \textbf{Where} \{ \\
		\null \hspace{0.8cm}\textit{?city1} ?x1 \textit{?city2}. \\
		\null \hspace{0.8cm}\textit{?city2} ?x2 \textit{?city3}. \}\\
	\endgroup
	}

This query is a collection of $2$ triple patterns, $\langle ?city1, ?x1,
?city2 \rangle$ and $\langle ?city2, ?x2, ?city3 \rangle$. The variables in a
triple pattern has prefix, question-mark $?$. 

All the answers to the query (along with their derivation polynomials) are listed in Table~\ref{tab:result}. The answers will be further filtered out to report only matches with probability $>0.5$.

\subsection{Query Evaluation on Probabilistic Dataset}

On a deterministic graph, a graph query result is a collection of projected
nodes of subgraphs that matches the query pattern. However, on a probabilistic
KG, the query engine has to perform additional task of computing probability of
each result item. This task is often referred to as \emph{probabilistic
inference}~\cite{dalvidb}.  Technically, \emph{probability inference} involves matching the
query pattern over all \emph{possible worlds} and computing the marginal 
probability of matches. The result of a query $Q$ over probabilistic KG is
given as,
\begin{align}
	A(Q) = \{ \langle m_{i}, Pr(m_{i}) \rangle, \dots, \langle m_{k}, Pr(m_{k}) \rangle \}
\end{align}
where each answer $m_{i}$ has a probability $Pr(m_{i})$ associated with it 
representing the overall probability of it being part of the answer set.

The number of possible world of a KG will be exponential in number
of edges, specifically $2^{|E|}$. Thus, the approach of enumerating all the
worlds and evaluating the query on each of them is impractical. Instead, each
answer is associated with a Boolean formula and the probability of this Boolean
formula to be \emph{true} over all assignments gives the probability of the
corresponding result item. This Boolean formula is referred to as \emph{lineage}.
We discuss how to compute lineage for query answers in Sec.~\ref{subsec:prov}.
For now, we continue with discussing the techniques to compute probability of
lineage.

A na\"ive technique to compute probability is to enumerate all
assignments of the Boolean formula and then count the satisfying
ones. This method of computing the probability is called
\emph{possible world computation}.  Clearly, the possible worlds
computation scales exponentially with the number of Boolean variables
in a formula.  The probability computation of a Boolean formula is
equivalent to the weighted model counting problem. Thus, probabilistic
inference is \emph{\#P-hard}.

Various heuristics are used to tackle this problem. These approaches
fall under the category of \emph{intensional query evaluation}.  One
of the most popular techniques used in intensional query evaluation is
based on converting a given Boolean formula into d-NNF forms, which
are known to be more tractable. This methodology is known as
\emph{compilation} \cite{compilation}. More details about different
kinds of compilation techniques are given in a popular survey by
Broeck and Suciu \cite{survey}. Tools such as \ctd
\cite{c2d}, D4 \cite{d4} are based on this compilation strategy.

\subsection{Lineage Computation}
\label{subsec:prov}

The provenance semiring model piggybacks the lineage computation on to the
query processing. Thus, the lineage is computed using the semiring $\left(PosBool(X_{\bullet}), \land,
\lor, 0, 1 \right)$ over positive Boolean expressions.
Tuples are annotated with independent Boolean variables $X_{1} ,\dots, X_{i}$.
For a query answer $m$, the lineage is a DNF (\emph{disjunctive normal form})
formula where each conjunct represents a possible derivation of $m$. A
conjunct is constructed by applying the AND operator ($\land$) on the Boolean
variables of the edges involved in that particular derivation.
For instance, in Table~\ref{tab:result}, the lineage of answer of the query,
shown in Fig.~\ref{fig:kg}, is given as $PosBool$. The Boolean formula
corresponding to the answer (SIN,MUN) has $2$ conjuncts $\left(e_{1} \land
e_{3}\right)$ and $\left(e_{2} \land e_{3}\right)$, each representing a
possible derivation of the answer. 

%% file: related-work.tex
\section{Related Work}
\label{sec:relwork}

\subsubsection*{Query Processing on Probabilistic Database}
\label{subsec:probWork}

The problem of query processing on probabilistic databases is quite
well-studied \cite{Cavallo, dichotomy, probabilisticGM, greenModel}.  A
fundamental result, aka the \emph{dichotomy} theorem, stating that either a
query can be computed in polynomial time or it is provably \#P-hard, was stated
in \cite{dichotomy}. 

Solutions addressing probabilistic inference can be categorized into two
classes: (a)~\emph{extensional}: those that compile the inference over
probabilistic data into a query plan, and (b)~\emph{intensional}: those that
directly manipulate the probability derivation expression of each tuple in the
output making use of its provenance.
Extensional techniques \cite{dalvidb, top-k, Dey, deshpande} are known to be more
efficient than intensional evaluation for the class of queries they can handle.
Extensional solutions cannot process arbitrary conjunctive queries. They cannot
handle \emph{self-joins} as the safe plan construction is based on the
assumption that the relations participating in the plan are independent. A self
join query violates this assumption.  Few practical systems such as
MystiQ~\cite{mystiq} and Orion~\cite{orion} have adopted extensional techniques
for simple queries. For hard queries, they rely on approaches based on
approximation algorithms~\cite{simulation,mcdb}.  

Event tables, which eventually become a standard way of data modeling of
intensional techniques, were introduced in~\cite{fuhr, esteban}. Green et
al.~\cite{provenancesemirings} proposed a generalized semiring model to
annotate probabilistic tuples with Boolean variables.
Trio~\cite{trio}, one of the early systems using intensional approach, relied
on data lineage to compute the probabilities of simple queries~\cite{trio_prob}
and Monte-Carlo simulation~\cite{simulation} for complex ones. Other notable
systems equipped to process queries over probabilistic database are
MayBMS~\cite{maybms} and SPROUT~\cite{sprout}.  Recently, Senellart et al.
presented \provsql~\cite{provsql}, a framework strictly adhering to the
semiring model~\cite{provenancesemirings, m-semirings} to represent
probabilistic data. For probabilistic inference, it supports three standard
techniques: possible world computation, knowledge
compilation~\cite{compilation} and Monte-Carlo simulation~\cite{simulation}.

\provsql can be considered as closest to our
system in terms of data modeling as we also use the semiring based framework
for data modeling. Unlike our focus on knowledge graphs, \provsql addresses the problem
for relational databases. 
Further, we equipped our solution with an efficient symbolic probability computation technique using a novel semiring. Note that our proposal can be easily incorporated into the \provsql framework. 

\subsubsection*{Query Provenance}

The capability of provenance to provide better insight into the query result
has resulted into various provenance aspects. These include \emph{why}
provenance~\cite{whyandwhere} explaining why an answer is part of a query
result, \emph{how} provenance~\cite{provenancesemirings} providing information
about the derivation process of an answer, and \emph{when}
provenance~\cite{whenprovenance} for tracking temporal data.  In this work, our
intent is to track the provenance (derivation process) of a query pattern match as well as
that of probability computation. Complying with the provenance semiring model \cite{provenancesemirings}, we propose a novel semiring which, unlike the $PosBool$
semiring, captures the \emph{how} provenance of the probability
\emph{computation} as well.

%% file: solution.tex
\section{Theoretical Framework}
\label{sec:framework}

We propose a framework to produce a symbolic expression corresponding to each
query answer. The evaluation of this expression with the edge probabilities
substituted will be the probability of the answer.  This is in contrast to
computation via PosBool \cite{provenancesemirings}. There is a layer of
probability computation starting from the PosBool expression~\cite{trio_prob,
fuhr} which is not a mere substitution of values.
The major attraction of semiring frameworks is that the construction of their
answer representations can be piggybacked onto query evaluation. Importantly,
our symbolic probability expressions lie in a semiring with a homomorphic
correspondence with PosBool.

We adopt the \emph{provenance semiring} model
\cite{provenancesemirings} to generate these expressions. Similar to
the $PosBool$ semiring, we annotate each edge $e_{i}$ with a random
variable $X_{i}$, indicative of the presence $(X_{i}=1)$ with
probability $p_i$ or absence $(X_{i}=0)$ with probability $(1 - p_i)$
of the edge. In our framework, probabilities of events in $PosBool$
are computed and stored as polynomials over indeterminates $\{p_i\}$.

In the $PosBool$ framework each derivation is associated
with the \emph{conjunct} of Boolean variables associated to the
involved edges. Here, the presence of each derivation $d_{i}$ of an
answer is interpreted as an event $D_{i}$ which, in turn, is defined by
the presence of edges involved in the derivation.  The presence of an
answer will rely on the presence of at least one of these
derivations. Thus, the probability of the answer is computed as
$Pr(\bigvee_{i} D_{i})$.  For instance, in Table~\ref{tab:result},
answer (SIN,MUN) has two derivations: $d_{1}= \{e_{1},e_{3}\}$ and
$d_{2}= \{e_{2},e_{3}\}$. The probability of event $D_{1}$ is
the product of the probabilities of edges $e_{1}$ and $e_{2}$ and that
of $D_{2}$ of $e_{2}$ and $e_{3}$. The probability of (SIN,MUN) being
a part of the result is computed as $Pr(D_{1} \vee D_{2})$.

An incremental step in computing the probability of a PosBool event can be
computing for a disjunct $Pr\left(E_1 \vee E_2\right) = Pr\left(E_1\right) +
Pr\left(E_2\right) - Pr\left(E_1 \wedge E_2\right)$ or for a conjunct
$Pr\left(E_1 \wedge E_2\right)$.

Notice that, in either case, computing $Pr\left(E_1 \wedge E_2\right)$
seems to necessitate keeping track of the exact events $E_1$ and
$E_2$. We show, however, that this is not necessary. Just the symbolic
probability expressions of the event $E_1$ and $E_2$ are enough for us
to compute probabilities incrementally. The standard edge independence
model is crucial for the correctness of this claim.

We start formally defining the semiring and presenting a semiring
homomorphism from $PosBool$ to our semiring.
Our domain of interest is the polynomial ring with integer
coefficients $Z[p_1, \ldots, p_n]$. We are interested in a subset of
this domain which we define as follows:

\begin{definition}[Flat Monomial; Flat Polynomial]
	We call a monomial \emph{flat} if it is a product of distinct variables.  A
	\emph{flat} integer polynomial is defined to be a sum of flat monomials.
\end{definition}

The function $flat$ flattens out a polynomial by reducing every
exponent greater than 1 to 1. Formally, for $k \in \mathbb{Z}$
  \begin{align}
    flat(k) &= k\\
    flat(f + g) &= flat(f) + flat(g) \\
    flat\left(c \cdot \prod_{i} p_{i}^{k_i}\right) &= c \cdot \prod_{i} p_{i} 
\end{align}
  where $k_i > 0$ for all $i, 0 < i \le n$ and $c \in \mathbb{Z}$

We use the notation $\overline{f}$ for $flat(f)$ here on. 

We propose two non-standard operators $\oplus$
and $\otimes$ on the set of flat integer polynomials. For flat
polynomials $f$ and $g$, the two operators are defined as follows:

\begin{itemize}
	\item $f \otimes g = \overline{f \times g}$
	\item $f \oplus g = f + g - (f \otimes g)$
\end{itemize}
where $+$, $-$ and $\times$ are the standard operations on
polynomials. Notice that the resultants of these operations are also
flat polynomials.

Let $Z_F[p_1, p_2, \ldots ,p_n]$ denote the set of flat integer
polynomials generated by $\{1,p_1,p_2,\ldots,p_n\}$ and the operators
$\otimes$ and $\oplus$. This is our domain of interest.

\begin{theorem}\label{thm:semiring}
  $(Z_F[p_1, p_2, \ldots ,p_n], \oplus, \otimes, 0,1)$ is a commutative semiring.
\end{theorem}

We defer the proof to Appendix \ref{appendix1}. 

Next, we define a map $H$ from $(PosBool(X_1, X_2, \ldots, X_n),$
$\lor, \land, 0, 1)$ to $(Z_F[p_1, p_2, \ldots, p_n], \oplus, \otimes,
0,1)$. This is done inductively on the structure of $PosBool$.

For $\theta_i, \theta_j \in PosBool(X_1,X_2,\ldots,X_n)$,
\begin{align}
	H(X_{i}) &= p_{i} \\
	H(\theta_{i} \land \theta_{j}) &= H(\theta_{i}) \otimes
          H(\theta_{j}) \\
	H(\theta_{i} \lor \theta_{j}) &= H(\theta_{i}) \oplus
          H(\theta_{j}) 
\end{align}

As defined, the above map is not yet guaranteed to be
well-defined. Elements of $PosBool$ obey relations beyond the usual
semiring axioms. For instance, consider that for arbitrary formulae
$E, F$ in $PosBool$, $E \lor (E \land F) = E$ (absorption), or that
$\lor$ also distributes over $\land$. These relations are precisely
why, for example, there is no canonical map from $PosBool$ to the
\emph{WHY}-provenance semiring~\cite{whyandwhere}. That $H$ is
well-defined would imply that these relations also hold in our
semiring $Z_F$. This is not obvious and turns out to be non-trivial to
prove directly. \comment {While we omit this proof and the proof of
  Theorem~\ref{thm:symprob} for lack of space,} The well-definedness
of $H$ and Theorem ~\ref{thm:symprob} are best proved in tandem.

\begin{theorem} \label{thm:symprob}
  For all formulae $E$ in $PosBool$
  $$ Pr(E) = H(E)\left(Pr(X_1), \ldots, Pr(X_n)\right) $$
  where the RHS denotes the flat polynomial $H(E)$ evaluated at $p_i = Pr(X_i)$ for $1 \le i \le n$.
  
\end{theorem}

This asserts that the probability of a formula $E$ in $PosBool$ is,
symbolically, the flat polynomial associated to it by $H$.

\subsection{Proof Sketch of Theorem~\ref{thm:symprob}}
\label{subsec:correctness}

For positive Boolean expressions $E_1$ and $E_2$ let $E_1
\overset{s.r}{\leadsto} E_2$ mean that $E_1$ can be reduced to $E_2$
using only semiring axioms. Given that $H$ is defined structurally
over $(\wedge, \vee)$ and that $Z_F$ is a semiring itself, it is clear
that
\begin{align}\label{srreduce}
  E_1 \overset{s.r}{\leadsto} E_2 \implies H(E_1) = H(E_2)
\end{align}

We reduce both the proof of well-definedness of $H$ and of
Theorem \ref{thm:symprob} to the special case of Theorem~\ref{thm:symprob} for formulae in DNF form. Notice that any
Boolean expression can be reduced to DNF form using only semiring
axioms (a simple inductive argument shows we only need left and right
distributivity of $\wedge$ over $\vee$).

For formulae $E_1$ and $E_2$
that are equivalent in $PosBool$, let $F_1$ and $F_2$ be formulae in
DNF form such that $E_1 \overset{s.r}{\leadsto} F_1$ and $E_2
\overset{s.r}{\leadsto} F_2$. Note that $E_1$, $E_2$, $F_1$ and $F_2$ are all
equivalent in $PosBool$. Consider the chain of equalities:
\begin{align}
  H(E_1) = H(F_1) \overset{?}{=} Pr(F_1) = Pr(F_2) \overset{?}{=}
  H(F_2) = H(E_2) \label{eqn:chain}
\end{align}
The middle equality follows from the equivalence of $F_1$ and $F_2$.
The equalities marked $\overset{?}{=}$ are the only ones that we have not
shown yet. These are invocations of Theorem~\ref{thm:symprob}
for formulae in DNF form. If these are established, we would have proved that
$H$ is well-defined ($H(E_1) = H(E_2)$) and also Theorem~\ref{thm:symprob} ($Pr(E_1) = Pr(F_1) = H(E_1)$). 

We, thus, proceed to prove Theorem~\ref{thm:symprob} for formulae in DNF
form (without assuming well-definedness of $H$).

\begin{lemma}
	\label{lemma:conjunct}
  For Boolean variables $\theta_i$ in $PosBool$ 
  $$ Pr(\bigwedge \limits_{i} \theta_{i}) = H(\bigwedge \limits_{i}
  \theta_{i}) $$
\end{lemma}

This is Theorem \ref{thm:symprob} for pure conjuncts.

\begin{proof}
  Let $I = \{\,i \mid X_i = \theta_j \text{ for some } j\,\}$. Recall
  that in a KG existence of different edges are independent
  events. Since this is what we are modeling all along, our random
  variables $X_i$ are also independent. Thus,
  $$Pr(\bigwedge \limits_{i} \theta_{i}) = \prod_{j \in I} p_j $$
  Also,
  $$H(\bigwedge \limits_{i} \theta_{i}) = \underset{i}{\otimes}\,
  H(\theta_i) = \underset{i}{\otimes}\, p_i= \prod_{j \in I} p_i$$ The
  last equality is a consequence of how flattening of polynomials works.
	\hfill{}
\end{proof}

\begin{lemma}
  For conjuncts of Boolean variables $C_i$
  
$$Pr(\bigvee \limits_{i} C_{i}) = H(\bigvee \limits_{i}
  C_{i})$$

  This is Theorem \ref{thm:symprob} for formulae in DNF form.
\end{lemma}

\begin{proof}
  We prove this by induction on the size of the disjunction. The base
  case is simply Lemma \ref{lemma:conjunct}. Consider,
 {\small
  \begin{align*}
    Pr\left(\bigvee \limits_{i = 1}^{n} C_i\right) &= Pr \left(\left(\bigvee \limits_{i = 1}^{n-1} C_i\right) \vee C_n\right) \\
    &= Pr \left(\bigvee \limits_{i = 1}^{n-1} C_i\right) + Pr(C_n) - Pr \left(\bigvee \limits_{i = 1}^{n-1} {\left(C_i\wedge C_n\right)}\right)
  \end{align*}
 }%
 Applying induction hypothesis
{\scriptsize \begin{align*}   
Pr\left(\bigvee \limits_{i = 1}^{n} C_i\right) &= H\left(\bigvee \limits_{i = 1}^{n-1} C_i\right) + H(C_n) - H \left(\bigvee \limits_{i = 1}^{n-1} {\left(C_i\wedge C_n\right)}\right) \\
	  &= H\left(\bigvee \limits_{i = 1}^{n-1} C_i\right) + H(C_n) - \left( H \left(\bigvee \limits_{i = 1}^{n-1} C_i \right) \otimes H\left(C_n\right) \right) \\
    &= H\left(\bigvee \limits_{i = 1}^{n-1} C_i\right) \oplus H(C_n)
= H\left(\bigvee \limits_{i = 1}^{n} C_i\right)
  \end{align*}
}%
\end{proof}%
Having proved Theorem~\ref{thm:symprob} for formulae in DNF form, we have the full equality chain in Eq.~\eqref{eqn:chain}.

%% file: algo.tex
\section{Algorithm}
\label{sec:algo}

\begin{table}[t]
	\tabcaption{Result of running example query shown in Fig.~\ref{fig:query}}
	\label{tab:result}
	\resizebox{\columnwidth}{!}{
		\begin{tabular}{llllll}
			\toprule
			\bf Result & \bf ProvPoly & \bf PosBool & \bf dervE & \bf symE & \bf Probability \\
			\midrule
			(DEL,BAR) & $e_{3}e_{4} $ & $e_{3} \land e_{4} $ & $e_{3} \otimes e_{4} $ & $e_{3}e_{4} $& $0.480$ \\
			(DEL,JKF) & $e_{3}e_{5} $ & $e_{3} \land e_{5} $ & $e_{3} \otimes e_{5} $ & $e_{3}e_{5} $& $0.360$ \\ 
			\multirow{2}{*}{(SIN,MUN)} & $e_{1}e_{3} +$ & $(e_{1} \land e_{3}) \lor$ & $(e_{1} \otimes e_{3}) \oplus$ & $e_{1}e_{3} + e_{2}e_{3}$ & \multirow{2}{*}{$0.564$} \\
			& $e_{2}e_{3}$ & $(e_{2}\land e_{3})$ & $(e_{2}\otimes e_{3})$ & $- e_{1}e_{2}e_{3}$ & \\
			\bottomrule
		\end{tabular}
	}
\end{table}

We extend our existing framework \huka~ \footnote{\url{https://github.com/gaurgarima/HUKA}} \cite{huka} that maintains query
	results along with their provenance for deterministic KG to probabilistic
	knowledge graphs. \huka supports provenance-aware query computation and
	result maintenance over deterministic dynamic KGs. It captures the
	\emph{how} provenance of query answers using the provenance semiring model.
	Our framework, \name, on the other hand, employs the novel semiring
	introduced in Sec.~\ref{sec:framework} for symbolically computing the
	answer probabilities.  We next discuss how to construct these symbolic
	probability expressions and, further, when a KG edge changes, how to
	maintain the expressions.

\subsection{Construction of Probability Expressions}
\label{subsec:symConstruction}

By the virtue of the semiring framework \cite{provenancesemirings}, the symbolic expression construction is piggybacked up the answer computation. Here, we discuss the construction of these expressions conceptually. 

Suppose we are given a graph query $Q$ of size $s$, i.e., there are $s$ triple patterns
in the query. We want to compute the query over a probabilistic KG $G$. The
result set $R$ of query $Q$ is given as,
\begin{align*}
	R = \{\langle a_{1}, \mathcal{D}_{a_{1}}\rangle, \dots, \langle a_{j},
	\mathcal{D}_{a_{j}}\rangle \}
\end{align*}
where, $\mathcal{D}_{a_{i}} = \{D_1 \dots D_l\}$ is the collection of all
derivations of an answer $a_{i}$.  Each derivation $D_{i}: \langle
e_{p},\dots,e_{q}\rangle$ encodes a subgraph, involving edges $e_{p}, \dots,
e_{q}$ that result in the corresponding answer. We will use
${e_{\bullet}}$'s both for the edges and the corresponding indeterminates in
our polynomial rings.

Similar to constructing a Boolean formula using the PosBool semiring, we first
construct \emph{derivation} expressions using our flat polynomial semiring.
From the mapping established in Sec.~\ref{sec:framework}, the conjuncts of a
Boolean formula corresponds to the $\otimes$ terms. Analogous to the
disjunction-of-conjunction form of Boolean formula, the derivation expression
is addition ($\oplus$) of multiplicative ($\otimes$) terms, i.e., $\oplus_{D
\in \mathcal{D}_{a_{i}}} ( \otimes_{e_{j} \in D} e_{j} )$ corresponds to
$\vee_{D \in \mathcal{D}_{a_{i}}} ( \wedge_{e_{j} \in D} e_{j} )$. For each
derivation $D_{k}$ of answer $a_{i}$, we multiply ($\otimes$) edges involved in
it, and iteratively add ($\oplus$) expressions of all the derivations of that
answer.

For instance, the answer (SIN,MUN), in Table~\ref{tab:result}, has two
derivations; thus, it has $2$ conjuncts and $2$ multiplicative terms in Boolean
formula (\textsf{PosBool}) and derivation expression (\textsf{dervE})
respectively. 
The PosBool formula and
derivation expression of all answers of Fig.~\ref{fig:query} is given in
Table~\ref{tab:result}.  Later, to compute the exact probability, existing
systems pass on the Boolean formula to compilation and counting tools, whereas
\name unfolds the derivation expression to get the equivalent probability
symbolic expression and evaluate it.

We handle one derivation at a time and incrementally construct the symbolic
expression. We iterate over all the derivations of an answer.  At iteration
$i$, we flatten out the polynomial constructed by multiplying edges involved in
derivation $D_{i}$.  Then we incrementally add ($\oplus$) this resultant flat
polynomial of $D_{i}$ to the symbolic expression computed at iteration $(i-1)$
to get the new updated symbolic expression. After exhausting all the
derivations, we get the final symbolic expression of the answer. To get the
concrete probability, we evaluate this expression by assigning
$Pr(e_{\bullet})$ to each variable $e_{\bullet}$. For our example, using the
edge probabilities shown in Fig.~\ref{fig:kg}, the symbolic expressions
\textsf{symE} and the resulting probabilities are shown in
Fig.~\ref{tab:result}.  Here, only the probability of pair (SIN,MUN) is above
the query threshold $\left(\geq 0.5\right)$.

The translation from \textsf{dervE} to
\textsf{symE} expressions involves a complete simplification of the $\otimes$
and $\oplus$ operators using their definitions. This is done \emph{before}
substituting concrete values in them. To see why this matters, consider the
derivation expression $(e_{1} \otimes e_{2}) \oplus (e_{1} \otimes e_{3})$ and
let each edge have probability $0.5$. If we \emph{first} assign values
(probabilities) to the indeterminates and then simplify the expression we get
$0.25 + 0.25 - 0.0625 = 0.4375$ as the concrete probability. However, the
correct value is $0.375$.

Our framework can also handle self-joins. In case of graph queries, the
situation arises when an edge satisfies more than one triple pattern of a
query. The multiplication operator ($\otimes$) ensures that the probability of
an edge is considered only once irrespective of the number of triple patterns
it satisfies in a derivation. 

\subsection{Maintainability of \name}
\label{subsec:maintain}

By virtue of \emph{incremental} symbolic expression construction procedure,
\name is capable of maintaining query answers under edge update operations by using
the following \emph{inverted} indexes:

\begin{itemize}

	\item An inverted index, \textsf{edgeToSymE}, that maps an edge to the
		collection of \emph{symbolic} expressions in which it participates.

	\item An inverse map, \textsf{symEval}, that associates a \emph{symbolic}
		expression to its current evaluation. This map essentially holds the
		current probability of all the answers.

\end{itemize}

We consider the following update operations: addition of an edge, deletion of
an edge, change of probability value of an edge.

\subsubsection{Addition of an Edge}
\label{subsubsec:edgeAddition}

\begin{table}[t]
\tabcaption{Updated query result after insertion of edge $e_{6}$.}
\label{tab:updatedResult}
\centering
\resizebox{\columnwidth}{!}{%
\begin{tabular}{lll}
	\toprule
	\bf Result & $\mathbf{symE}$ & \bf Probability \\
	\midrule
	(DEL,BAR) & $e_{3}e_{4} + e_{6}e_{4} - e_{3}e_{4}e_{6}$ & $0.5440$ \\ 
	(DEL,JFK) & $e_{3}e_{5} + e_{6}e_{5} - e_{3}e_{5}e_{6}$ & $0.4080$ \\
	(SIN,MUN) & $(e_{1} e_{3} + e_{2}e_{3} - e_{1}e_{2}e_{3}) \oplus (e_{1} \otimes e_{6}) \oplus (e_{2} \otimes e_{6}) $  & $0.6395$ \\
	\bottomrule
	\end{tabular}
	}
	\vspace{-6mm}
\end{table}

A newly added edge may result in generation of either new answers or more
derivations of an existing answer.  A new \emph{answer} generation is a simple
case of symbolic expression construction mentioned in
Sec.~\ref{subsec:symConstruction}. However, the accommodation of new
\emph{derivations} involves updating symbolic expressions and, thus, the
resulting probabilities of \emph{affected} answers. We use the existing system
\huka \cite{huka} to efficiently identify and supply new derivations to our
\name framework. Here, we focus on updating the affected symbolic expressions.

 Suppose, we add an edge $e_{6}(DEL,MUN,A_{1},0.2)$ to
the example KG shown in Fig.~\ref{fig:kg}. It would result in $2$ new
derivations, $D_{1}: \langle e_{1}, e_{6} \rangle$ and  $D_{2}: \langle e_{2},
e_{6} \rangle$ of answer \textit{(SIN,MUN)}. \name \emph{incrementally} updates the existing symbolic 
expression of an \emph{affected} answer by absorbing one
derivation at a time. With existing expression $e_{1}e_{3} + e_{2}e_{3}-e_{1}e_{2}e_{3}$, the updated expression would be $(e_{1}e_{3} + e_{2}e_{3}-e_{1}e_{2}e_{3}) \oplus e_{1}e_{6} \oplus e_{2}e_{6}$. The updated expressions and the probabilities of all the answers is reported in
Table~\ref{tab:updatedResult}.

Often queries specify a probability threshold to find answers with
probabilities above a certain threshold. In our example query, suppose the user
wants flight routes with overall probability greater than $0.5$. Thus, instead
of computing the exact probability after each edge insertion, we require a
mechanism to quickly filter out answers that cannot pass the threshold, after a
new derivation.  To this end, we devise a simple method based on the following observation. 

\begin{observation}
	\label{obs:optimize}
	An \emph{upper bound} on the probability computed by adding two symbolic
	expression $symE_{i}$ and $symE_{j}$ is
	{
	\begin{align*}
		Pr &\left(symE_{i} \oplus symE_{j} \right) \\
		&\leq Pr(symE_{i}) + Pr(symE_{j}) - Pr(symE_{i}) \times Pr(symE_{j})
	\end{align*}
	}
\end{observation}

This observation gives an upper bound on the new probability value of an
answer.  This upper bound that entirely ignores the dependence among the
different derivations of an answer, referred to as propagation score~\cite{dissociation}, is often used for approximating the exact probability of answers.  The
upper bound computation is quite efficient as it involves straightforward
arithmetic computations.

Thus, for each \emph{affected} answer, \name follows a filter-and-refine
mechanism: it first computes the upper bound, and then decides if the updated
symbolic expression computation is needed at that point in time. For instance,
after the insertion of edge $e_{6}$ (probability $0.2$), the upper bound of
answer (DEL,BAR) and (DEL,JFK) are $0.563$ and $0.437$ respectively.  Thus, we
compute the exact probability of \emph{only} (DEL,BAR) as it passes the query
threshold $0.5$.

To ensure correct operation of \name to handle future edge update requests, we
cannot altogether avoid the updated symbolic expression computation.  We
defer it to improve the response time, i.e., the time taken to report the
updated answers of \emph{affected} queries. 

\subsubsection{Edge Probability Update}
\label{subsubsec:edgeProbUpdate}

An edge probability update operation affects \emph{only} the answer
probabilities and not their symbolic expressions (assuming the new probability
to be non-zero).  Thus, instead of re-evaluating the symbolic expression of an
affected answer from scratch, \name computes only the offsets corresponding to
the new probability.
Suppose, the probability of edge $e_{i}$ is updated from $Pr(e_{i})$ to
$Pr(e_{i})'$. First, \name fetches entry of $e_{i}$ in \textsf{edgeToSymE} to
get all the symbolic expressions in which $e_{i}$ participated.  Then, for each
such symbolic expression $symE_{j}$, it computes the offset. The \emph{offset}
is calculated by re-evaluating the monomials of the symbolic expression in
which $e_{i}$ appeared.  For instance, if the probability of $e_{2}$ is updated
to $0.6$, then symbolic expression of (SIN,MUN) would be fetched from
\textsf{edgeToSymE}. The symbolic expression has $2$ relevant monomials,
$e_{2}e_{3}$ and $e_{1}e_{2}e_{3}$, as shown in Table~\ref{tab:result}.  The
\emph{offset} is calculated as sum of valuations of all the relevant monomials.
These monomials are evaluated by assigning corresponding edge probabilities to
all involved edges except $e_{i}$. The updated edge $e_{i}$ variable is
substituted by value $Pr(e_{i})-Pr(e_{i})'$. Here, the relevant monomials
$e_{2}e_{3}$ and $e_{1}e_{2}e_{3}$ are evaluated to $0.06$ and $0.048$
respectively, with $e_{1}=0.8$, $e_{2} = 0.1$ ($Pr(e_{2})-Pr(e_{2})' = 0.1$)
and $e_{3}=0.6$.  Thus, the offset is $0.108$ ($= 0.06+0.048$). Finally, the
updated probability value the answer is given as \textit{newP} = \textit{oldP}
-- \textit{offset}, where \textit{oldP} is the probability of the answer before
the update.

\subsubsection{Deletion of an Edge}
\label{subsubsec:edgeDeletion}

When an edge gets deleted, the derivations in which it participated becomes
\emph{invalid} as they cannot generate the corresponding answer anymore. There
are two ways of computing the updated symbolic expression of $a_{i}$.  Suppose
$k$ out of $d$ derivations of an answer $a_{i}$ become \emph{invalid}.  The
first one involves iterating over the monomials of the current symbolic
expression and dropping off the monomials containing the deleted edge. The
other method involves removing the $k$ invalid derivations from the current
derivation list and recomputing the symbolic expression with this updated list
of derivations. As we will see in Sec.~\ref{subsec:analysis}, the number of
monomials in a symbolic expression constructed from $d$ derivations can be
exponential in $d$. When that bound is attained, the time involved in
manipulating an expression is as much, or more, than that for recomputing it
from the scratch, i.e., $O(2^{d-k})$. Therefore, we adopt the latter strategy.

%% file: analysis.tex
\subsection{Analysis}
\label{subsec:analysis}

\begin{table}[t]
	\tabcaption{Parameters of a Query}
\centering
	\resizebox{\columnwidth}{!}{%
	\begin{tabular}{ll}
		\toprule
		\bf Parameter & \bf Description \\ 
		\midrule
		$s$ & Size of query, i.e., number of triple patterns \\
		$d$	& Number of derivations of an answer \\
		$n$ & Total number of edges across all $d$ derivations of an answer \\
		$m$ & Number of monomials in a symbolic expression (flat polynomial) \\ \bottomrule
	\end{tabular}
	}
\label{tab:symbols}
\end{table}

We start by introducing the different parameters that characterize the queries and
their answers. We refer to the number of triple patterns in a query as
\emph{query size}, $s$.
\comment{Our running query is of size $s=2$.}
An answer
$a_{k}$ of an query can have $d$ derivation terms, each denoted as $D_{i}$, $1 \leq
i \leq d$.
Further, let the number of edges involved across all $d$ derivations be
$n$. \comment{Here, in our example, answer \textit{(SIN,MUN)} has $2$ derivations which
involves a total of $3$ edges, i.e., $d=2$, $n=3$.}  On simplifying the
derivation expression $dervE_{a_{k}}$ of $a_{k}$, we get the corresponding
symbolic expression $symE_{a_{k}}$. The number of monomials in the symbolic
expression (\emph{flat} polynomial) is denoted as $m$.  \comment{Effectively, a symbolic
expression is a flat polynomial over $n$ variables containing $m$ monomials.
\comment{For instance, the symbolic expression $e_{1}e_{3} + e_{2}e_{3} -
e_{1}e_{2}e_{3}$ of answer \textit{(SIN,MUN)} is a flat polynomial over
$\{e_{1},e_{2},e_{3}\}$ and has $3$ monomials.}} 
Table~\ref{tab:symbols} summarizes the parameters.

The cost of construction of symbolic probability expressions involves
simplifying a derivation expression to a polynomial in $Z[e_1,\ldots,e_n]$ by
using the definitions of semiring operations $\otimes$ and $\oplus$. 
Let $f$ and $g$ be two flat polynomials and $|\bullet|$ be the number
of distinct monomials in the simplified form of polynomial
$\bullet$. Then, the bound on the number of distinct monomials is
{
\begin{align}
  |f\otimes g| &\le |f|\times|g| \\
  |f \oplus g| &\le |f| + |g| + (|f|\times|g|) \label{eq:plussing}
\end{align}
}
The time taken to compute $f \times g$ is $O(|f| |g| (\log|f|+\log|g|))$, and
that for $f + g$ is $O(\min\{ |g| \log|f|, |f| \log|g| \})$.  The logarithmic
terms correspond to searching for new entrants among existing monomials.

We compute the symbolic probability expressions incrementally by adding
terms corresponding to one derivation at a time. In other words, at
step $i$, derivation $D_{i}$ is added ($\oplus$) with the resultant
symbolic expression of step $i-1$.  If $m_{i}$ is the size (number of
monomials) of the symbolic expression obtained after the $i^\text{th}$
step, we have the recursive bound
$m_{i+1} \le 2 \cdot m_i + 1$
with $m_1 = 1$. Therefore,
$m_d \le 2^d$.

Notice that each update step can be analyzed as a
special case of Eq.~\eqref{eq:plussing} with $|g| = 1$.  Given the bound
on $m_d$, the update step ($(d+1)^\text{st}$ step) takes $O(d \cdot 2^d)$
time.  Thus, the total computation time is also $O(d \cdot 2^d)$ in the
worst case.

Interestingly, we can get other bounds on $m_d$ as well. The symbolic
computation is a flat polynomial on $n$ variables. Therefore, it can have at
most $2^n$ monomials. The update step takes $O(n \cdot 2^n)$ time and the total
computation time is $O(d \cdot n \cdot 2^n)$ in the worst case.

Towards another bound, let us view monomials as sets of
variables. After either semiring operation, $\otimes$ or $\oplus$,
every monomial in the resultant polynomial is a superset of a monomial
from the polynomials operated upon. Therefore, every monomial in our
symbolic probability expression is a superset of a monomial
corresponding to a single derivation. In cases where the smallest
number of variables in a derivation (say, $t$) is large compared to
the total number of variables, the bound on the total number of
monomials, $\Sigma_{i=t}^n{{n}\choose{i}}$, is much better than
$2^n$. In case where $n-t$ is very small, this bound degenerates to
$O(n^{n - t})$ and the update step time is $O(n^{n - t + 1})$. To
locate answers where this bound takes effect, we found that using $s$
as a proxy for $t$ works well in practice.

In the case where these $n$-dependent bounds are smaller than the
$d$-dependent bounds, the update time is markedly lesser compared to
the total compute time.  For example, if $n-t$ is very small and $n$
does not change with incoming derivations, the time taken for each
subsequent update step is $O(n^{n - t + 1})$.  This becomes a
successively smaller fraction of the total computation time as more
derivations come in. This is in contrast to exponential time update
steps, where the time taken for each subsequent update step is the
same fraction of the total compute time till that point. These are the
answers where we expect the most advantage from our maintenance
algorithm.

%% file: experiment.tex
\section{Experiments}
\label{sec:expts}

Query processing over probabilistic KG involves two tasks: finding the answers
	of a posed query, and computing the probability of each answer of the query
	along with its provenance. Many systems such as \huka \cite{huka},
	\tripleprov \cite{tripleProv}, \provsql \cite{provsql} compute the graph
	query along with provenance polynomials over deterministic data. Our \name
	framework can be plugged into any of them.  We have used \huka as the base
	system.  In this section, therefore, we focus on the performance of the
	probability computation task starting from the derivation lists of the
	query answers.  There is another important dimension of \name, that of
	maintainability since it is quite common for KGs to undergo changes.
	Therefore, we also test our system on maintenance time under these
	operations against the complete re-computation cost of the symbolic
	expressions.

\subsection{Setup}

\subsubsection*{Datasets}

We consider two widely used benchmark datasets in our experimental evaluations:
%
%
		(a)~\textbf{YAGO2} \cite{yago2,hoffart2011yago2}, an automatically
		built ontology gathering facts from different sources like
		Wikipedia,
		GeoNames,
		etc. It has
		$\sim 23$\,M facts over $\sim 5.8$\,M real-world entities.
%
		(b)~\textbf{gMark} \cite{gmark}, a synthetic dataset generated by
		a schema-driven data and workload generator, \gmark. We used the schema
		of LDBC SNB~\cite{ldbc}
		to generate a graph with $0.9$\,M nodes and $2.2$\,M edges.
%

\subsubsection*{Query Collection}

For the \yago dataset, we used a set of queries on which the RDF-3X was
originally validated~\cite{neumann2010rdf}. We chose $3$ out of $6$ benchmark
queries since the other queries have answers with only a single derivation. The
probability of answers with a single derivation can be computed simply by
multiplying the probabilities of edges involved. This is a corner case that
does not serve to make any comparison.  The $3$ chosen queries are fairly large
and complex, and have $7.25$ triple constraints on an average.
%
For the \gmark dataset, we generated queries of size between $3$
and $7$. We generated $100$ queries out of which $11$ have answers with
multiple derivations. The average query size of these $11$ queries is
$4.24$ triples.

The statistics of the datasets are reported in Table~\ref{tab:yagoQ} and
Table~\ref{tab:gmarkQ}.  Interestingly, given the size of \gmark KG, queries of
even size $7$ have quite large ($\approx 6000$) answer sets. This is due to the
fact that none of the generated queries have bound variables, i.e., for all the
triple patterns, both subject and object are variables.  Since each answer is
dealt with independently, the variation across this large number of answers
helps us to evaluate the methods thoroughly.

\subsubsection*{Implementation}

We conducted all our experiments on a 32-core 2.1GHz CPU,
512GB RAM machine with 1TB hard drive. Our implementation is single-threaded in Java. Our codebase is publicly available \footnote{\url{https://github.com/gaurgarima/HaPPI}}. The KGs \yago and \gmark are realized as Neo4j property graphs of size $11$\,GB and $211$\,MB respectively. The inverted indexes
\textsf{edgeToSymE} and \textsf{symEval} are of size $2.4$\,MB and $2.3$\,MB respectively for \yago, and $805$\,MB and $804$\,MB
respectively for \gmark. 

\subsubsection*{Choice of Baseline Systems}

For baseline system selection, we focused on only \emph{intensional} technique
based systems since \emph{extensional} techniques cannot handle
\emph{self-joins} (Sec.~\ref{subsec:probWork}). We chose a recently proposed
system \provsql~\cite{provsql} that implements \textsf{PosBool} semiring based
probabilistic database on top of PostgreSQL.
It relies on $3$ standard ways to compute the probability: (a)~possible world
computation, (b)~knowledge compilation, and (c)~Monte Carlo technique.  We have
not considered the Monte Carlo approach as, unlike the other two approaches, it
computes approximate probabilities.  
We adopted the implementation of possible world (\pw) and knowledge compilation
(\ctd) used in publicly available \provsql as our baselines.

\begin{table}[t]
	\centering
	\tabcaption{Details of answer set of \yago queries.}
	\footnotesize{
		\begin{tabular}{ccccccc} 
			\toprule
			\bf Query & \bf Query & \bf Size of & \multicolumn{2}{c}{\bf \#Derivations, $d$} & \multicolumn{2}{c}{\bf \#Distinct edges, $n$}  \\ 
			\bf Id & \bf Size, $s$ & \bf Answer Set & min & max & min & max \\
			\midrule
			$Q_{1}$ & 6 & 2 & 2 & 2 & 9 & 9  \\
			$Q_{3}$ & 4 & 728 & 2 & 14 & 6 & 30  \\
			$Q_{6}$ & 6 & 544 & 2 & 5 & 8 & 14  \\
			\bottomrule
		\end{tabular}
	}
	\label{tab:yagoQ}
\end{table}

\begin{table}[t]
	\tabcaption{Details of answer set of \gmark queries.}
	\footnotesize{
		\begin{tabular}{ccccccc} 
			\toprule
			\bf Query & \bf Query & \bf Size of & \multicolumn{2}{c}{\bf \#Derivations, $d$} & \multicolumn{2}{c}{\bf \#Distinct edges, $n$} \\ 
			\bf Id & \bf Size, $s$ & \bf Answer Set & min & max & min & max  \\
			\midrule
			$Q_{4}$ & 7  & 6048 & 5 & 15 & 8 & 9  \\
			$Q_{6}$ & 4 & 1536 & 2 & 87 & 5 & 69  \\
			$Q_{9}$ & 5 & 88 & 2 & 25 & 7 & 52  \\
			$Q_{21}$ & 5 & 411 & 6 & 126 & 6 & 9  \\
			$Q_{23}$ & 3 & 2325 & 2 & 122 & 4 & 124  \\
			$Q_{32}$ & 3 & 32 & 2 & 2 & 6 & 6  \\
			$Q_{35}$ & 3 & 10 & 2 & 2 & 4 & 5  \\
			$Q_{38}$ & 3 & 4929 & 2 & 257 & 4 & 620  \\
			$Q_{46}$ & 7  & 6156 & 2 & 24580 & 8 & 8681  \\
			$Q_{54}$ & 6 & 400 & 2 & 6 & 7 & 16 \\
			$Q_{90}$ & 6 & 62 & 4 & 35 & 7 & 10 \\
			\bottomrule
		\end{tabular}
	}
	\label{tab:gmarkQ}
\end{table}

\subsection{Probability Computation Time}
\label{sec:computation}

The \textit{probability computation time} is the time taken to compute the
probability of query answers given a list of their derivation(s).  Variance in
probability computation time across query answers happens due to the following
characteristics: (i) $d$, the number of derivations, (ii) $n$, the total number
of edges involved, and (iii) $s$, the size of query.  We call this triple $(d,
n, s)$ the \emph{answer signature}.

\subsubsection*{Performance across Answer Signatures}

Since answer signatures across a query show wide variations,
we try to understand the trends by grouping the answer signatures
into buckets.  We expect, from
Sec.~\ref{subsec:analysis}, that \name scales exponentially
with $d$. Similarly, \pw scales exponentially with $n$.  Hence,
query answers are first grouped on the basis of
$d$, and then grouped further on the basis of $n$.
The bucket boundaries are chosen such
that the variation within a bucket is not very high.
Table~\ref{tab:baselineyago} and Table~\ref{tab:baselinegmark} show
the detailed results across these buckets. 
The \emph{count} column shows the number of query answers each
signature bucket has.
While the
performance of \name can also depend on query size $s$, this
dependence shows up for a very small set of answers (when $s$ is
almost as large as $n$).  Hence, we do not show $s$ in the tables.


For \yago queries, it can be seen that \name massively outperforms both the
systems.  The largest absolute time for any derivation is only $130\mu$s for
\name.  This is because $d$ and $n$ are not very large.  As expected, \pw shows
an exponential scaling with $n$.  \ctd shows a flat trend across buckets but a
high variation within them. Knowledge compilation based techniques are
sensitive to the precise Boolean formula whose probability is being computed
and not just to its size parameters. \comment{While there should still be a
general trend of computation time increase with size of answer parameters, this
asymptotic trend does not show up within the range of our data.}
The trend across \gmark queries is more interesting.  Up to $d = 12$ and $n <
18$, \name performs very well.  When $n \geq 18$ (for $d \in [9,12]$), the time
for \name shoots up to more than $4$ms.  While \pw could not finish even after
$30$s, \ctd is faster than \name in this range.  When $d$ is even larger ($>
13$), the time for \name jumps to $41$ms.  \ctd remains more or less constant
in the range of $4$ms.


\begin{table}[t]
	\tabcaption{Probability computation times ($\mu$seconds): \yago.}
	\resizebox{\columnwidth}{!}{
		\begin{tabular}{cc|c|c|c|c}
			\toprule
			$\mathbf{d}$ & $\mathbf{n}$ & \bf Count &	\bf \name	&	\bf \pw	&	\bf \ctd 	\\
			\midrule
			$[2,5)$ & $-$ &	1,225 &	19.22 $\pm$	11.07 &	444.72 $\pm$	518.29 &	2,720.25 $\pm$	1,723.80 \\
			\cmidrule(lr){1-6}			
			\multirow{2}{*}{$[5,9)$} & $< 13$ &	7 &	84.61 $\pm$	35.62 &	2,624.87 $\pm$	359.38 &	1,709.67 $\pm$	530.58 \\
			& $\geq 13$ &	39 &	130.59 $\pm$	87.01 &	19,265.49 $\pm$	19,102.84 &	2,233.88 $\pm$	1,636.08 \\
			\bottomrule
		\end{tabular}
	}
	\label{tab:baselineyago}
\end{table}

\begin{table}[t]
	\centering
	\tabcaption{Probability computation times ($\mu$seconds): \gmark.}
	\resizebox{\columnwidth}{!} {
		\begin{tabular}{cc|c|c|c|c}
			\toprule
			$\mathbf{d}$ & $\mathbf{n}$ & \bf Count &	\bf \name	&	\bf \pw	&	\bf \ctd 	\\
			\midrule
			\multirow{2}{*}{$[2,5)$} & $< 6$ & 4,558 & 5.69 $\pm$ 3.05 & 10.88 $\pm$ 5.74 & 4,811.58 $\pm$ 3,200.32 \\
			&$\geq 6$& 4,321 & 16.95 $\pm$ 10.44 & 1,564.73 $\pm$ 3,843.17 & 3,827.67 $\pm$ 3,101.51 \\ 
			\cmidrule(lr){1-6}			
			\multirow{2}{*}{$[5,9)$} & $< 13$ & 5,523 & 87.67 $\pm$ 64.81 & 900.61 $\pm$ 918.17 & 4,202.18 $\pm$ 2,784.27 \\ 							
			& $\geq 13$ & 673 & 249.31 $\pm$ 171.56 & 335,246.76 $\pm$ 738,398.70 & 3,367.52 $\pm$ 2,834.69 \\
			\cmidrule(lr){1-6}			
			\multirow{2}{*}{$[9,13)$} & $< 18$ & 360 & 265.38 $\pm$ 183.52 & 353,521.69 $\pm$ 761,014.83 & 3,418.17 $\pm$ 2,890.34 \\
			& $\geq 18$ & 439 & 4,684.11 $\pm$ 3,300.08 & time-out & 3,758.99 $\pm$ 2,937.08 \\ 
			\cmidrule(lr){1-6}			
			\multirow{1}{*}{$[13,16)$} & $-$ & 2,415 & 41,697.01 $\pm$ 32,091.19 & time-out & 3,924.66 $\pm$ 3,176.96 \\
			\bottomrule						
		\end{tabular}
	}
	\label{tab:baselinegmark}
\end{table}

\begin{figure}[t]
	\centering
	\includegraphics[scale=\figscale]{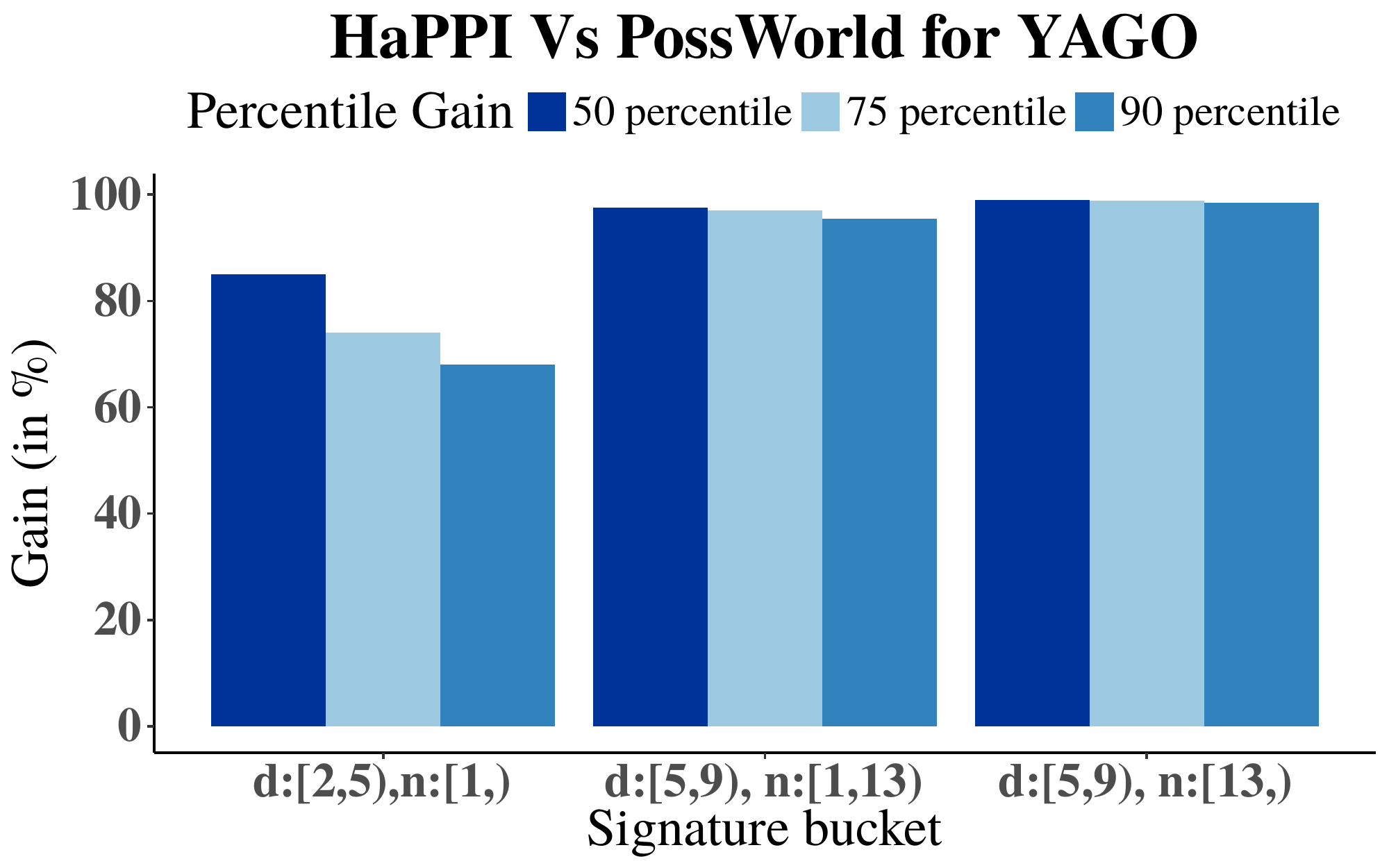}
	\includegraphics[scale=\figscale]{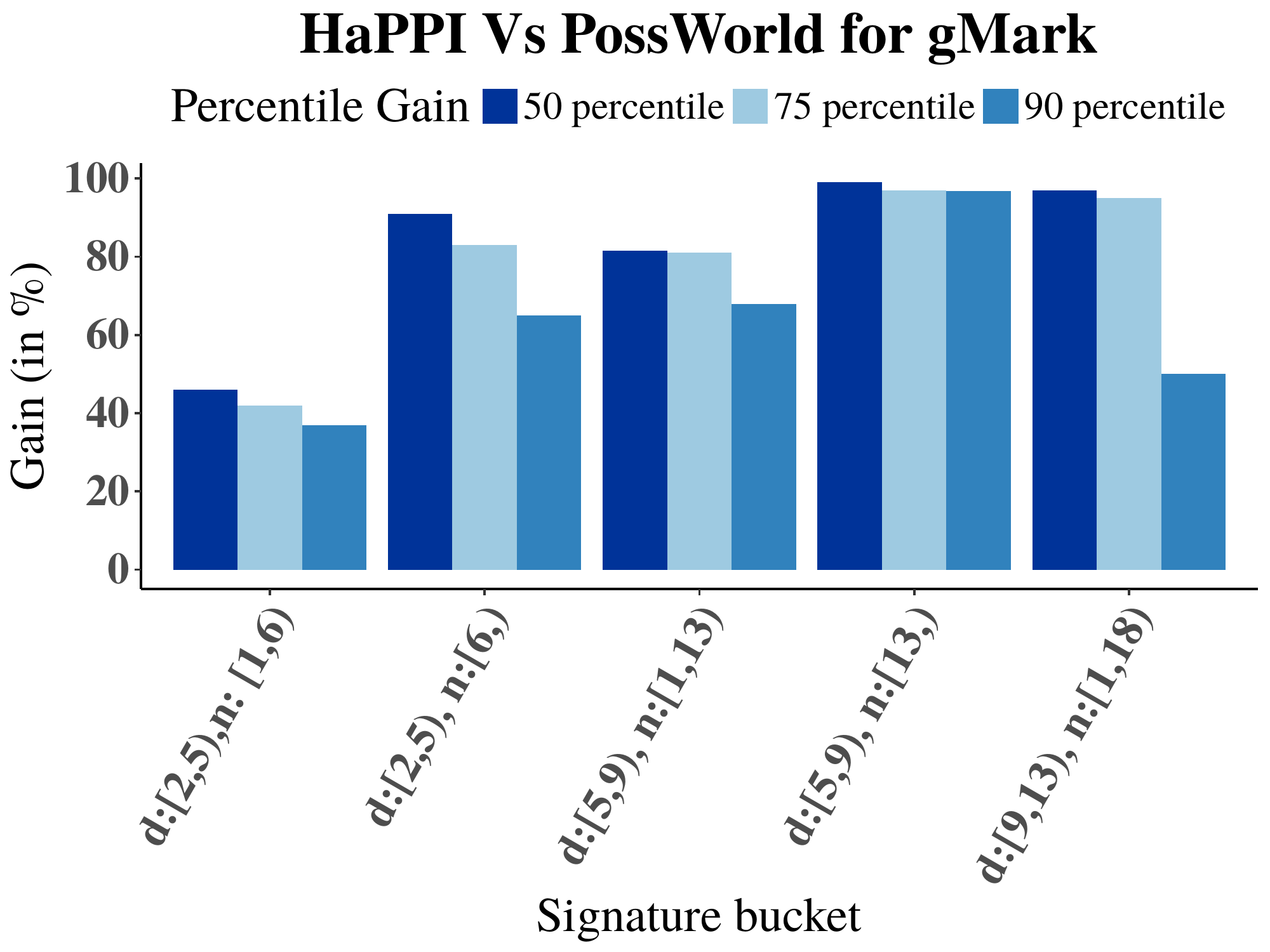}
\figcaption{Percentile gains of \name against \pw.}
\label{fig:percentilePW}
\end{figure}

\begin{figure}[t]
	\centering
	\includegraphics[scale=\figscale]{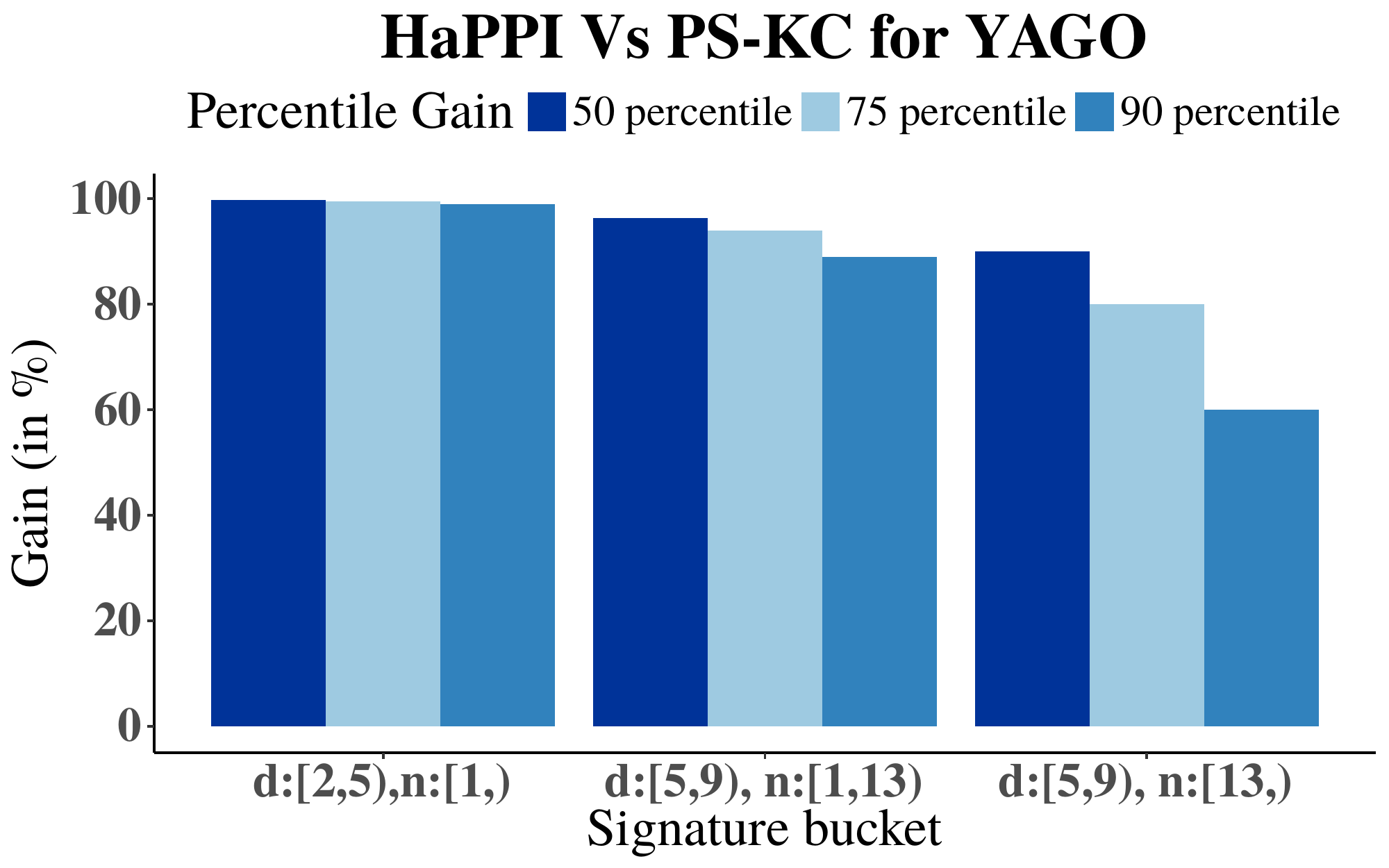}
	\includegraphics[scale=\figscale]{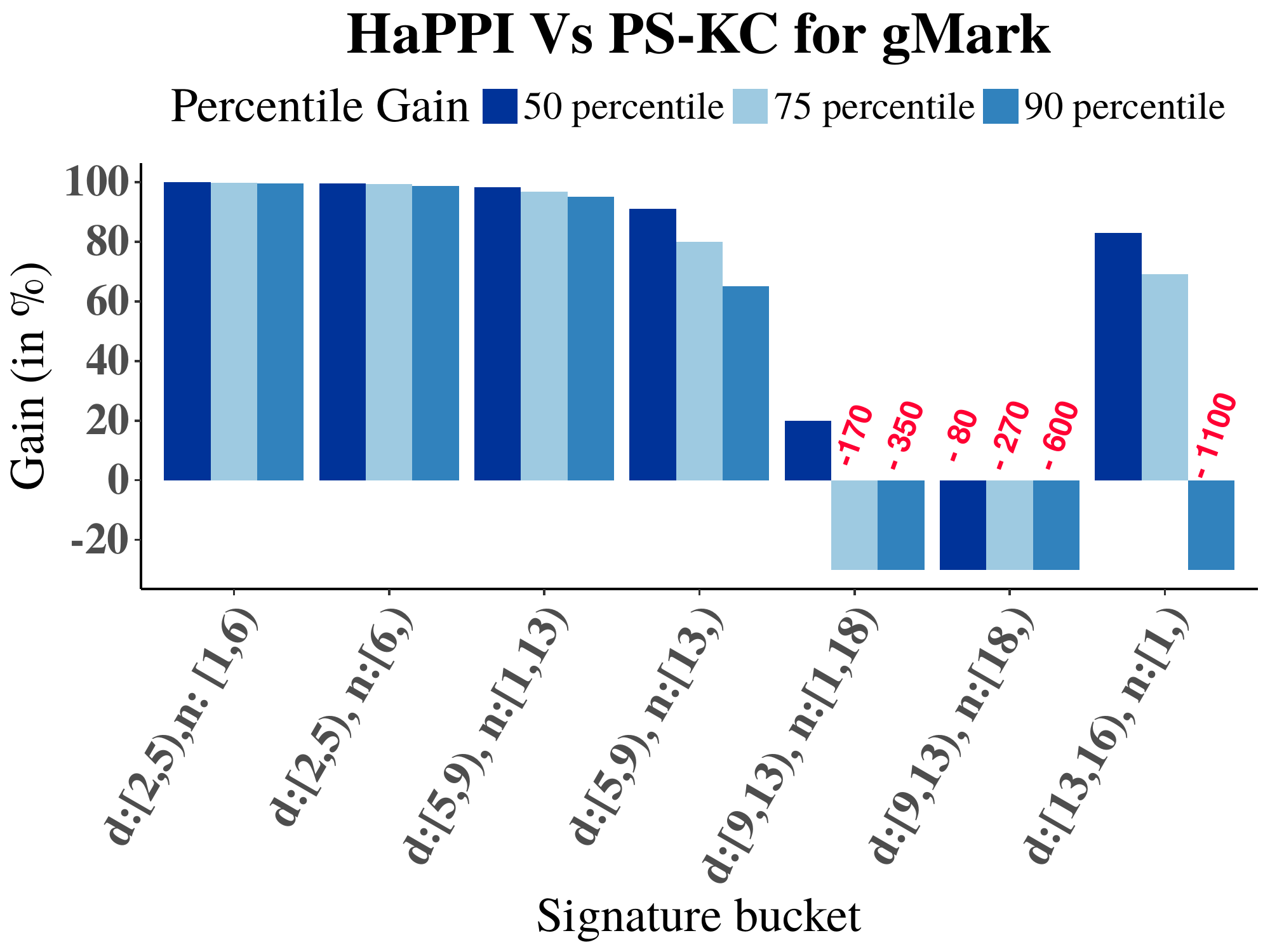}
	\figcaption{Percentile gains of \name against \ctd.}
	\label{fig:percentileC2D}
\end{figure}

\subsubsection*{Percentile Gains}

We dig deeper to analyze the gains of \name over the other two systems in more
detail. Timing averages of exponential systems tend to be dominated by corner
cases. To guard against that, and to get a better insight, we employ a new
metric for comparison based on \emph{percentile gains}.  By reporting a
\emph{$p^\text{th}$ percentile gain} of $x\%$, we mean that for $p\%$ of the
answers, \name took at least $x\%$ lesser time than the method compared
against.  The gain is computed using the following ratio: 
%
$$\text{gain}_\text{method} =
(\text{time}_\text{method}-\text{time}_\name)/\text{time}_\text{method}$$ 
For instance, Fig.~\ref{fig:percentilePW} reports a $50^\text{th}$ percentile
gain of $82\%$ for \name over \pw for the bucket $d \in [2,5)$, $n \geq 1$.
This means that for $50\%$ of the answers in this bucket, the probability
computation time of \name is at most $100 - 82 = 18\%$ of the corresponding
time of \pw.  A negative gain indicates that \name is slower.  \comment{ For
example, if for a bucket having $100$ answers, the $75^\text{th}$ percentile
gain of \name is $80\%$, this means that for $75$ answers, the probability
computation of \name is at most $100 - 80 = 20\%$ of the corresponding time of
the compared method.  Note that if \name is slower, gain is negative.  }
\name shows superior performance across all signature buckets. \comment
{\ab{The gains displayed are lesser than those the averages suggest. This is to
be expected, because timing averages of exponential systems tend to be
dominated by the corner cases. However, the sizable gains even for the $90$th
percentile showcase consistent and marginal better performance for individual
answers.} \ab{Not clear -- can we cut?}}
Fig.~\ref{fig:percentileC2D} shows that the better performance of \name over
\ctd for small answer signatures is uniform inside the buckets with
little deviation even for corner cases. The asymptotic deterioration of
\name is also visible. There is an exception to this trend. We discuss
the unexpected gains in the bucket $d\in[13,16)$, $n \geq 1$ in
Sec.~\ref{sec:adaptive}.

\subsection{Adaptive Framework}
\label{sec:adaptive}

While \name outperforms \ctd and \pw significantly for lower ranges of $d$ and
$n$, it is quite slow for larger values.
Since \ctd takes roughly the same amount of time to compute probabilities for
the very smallest of query answers to the largest, an \emph{adaptive} strategy
involving both \name and \ctd seems to be the best.  The \emph{adaptive}
strategy utilizes the best of both the worlds: it employs \name for lower
ranges of $d$ and $n$, and switches to \ctd when these parameters become large.
We report $3$ bounds, $O(2^{d})$, $O(d \cdot 2^{n})$ and
$O(n^{s-n})$, on the probability computation of \name
(Sec.~\ref{subsec:analysis}). We, thus, expect to use \name for small values of
$d$, $n$ or $s-n$.

To notice the $s-n$ dependence, we need to look at higher values of $d$ and
$n$, since \name would outperform \ctd anyway for smaller values.  For the
range $13 \le d < 16$, this is precisely what happens.  Closer inspection
showed that $2089$ out of $2415$ answers in this range had signatures $(5, 21,
7)$ or $(7, 15, 9)$ and \name outperformed \ctd for all these answers.  This
explains the aberrant percentile gain of \name. It also empirically validates
our analysis on $s-n$ values. 

We now elaborate on our \emph{adaptive} strategy.  For each answer, we
use \name if one of these $3$
conditions is satisfied: (a)~$d < 10$, (b)~$n < 8$, or (c)~$n - s < 3$.
\comment{
\gaur{change the enumeration to inline}

\begin{enumerate}[label=(\alph*)]
	\item $d < 10$, or
	\item $n < 8$, or
	\item $n-s < 3$.
\end{enumerate}
}
We refer to this signature range as the \name domain of answers.
We employ \ctd outside this domain. 

The overall probability computation time of \emph{adaptive} (\name/\ctd) and
pure \ctd techniques for \yago and \gmark queries are reported under column
`Computation Time' in Table~\ref{tab:yagooverall} and
Table~\ref{tab:gmarkoverall} respectively.  We have also reported the
percentage of query answers for which \name was used in the adaptive system.
For \yago queries, our adaptive approach is on an average $303$ times faster
than pure \ctd.  For \gmark queries, we record an average speed-up of $317$
times. The speed-up is, as expected, very high for queries $Q_{32}, Q_{35},
Q_{54}$, where \name is used for all the query answers.  Importantly, the
adaptive system gives a $21\%$ speed-up for even $Q_{46}$ where \name is used
for only about $37\%$ of the answers and the answer signatures go up to
$d=24580$ and $n=8681$.



\begin{table}[t]
\tabcaption{Performance of adaptive system (in $\mu$seconds) for query computation and maintenance for \yago query set.}
	\centering
	\resizebox{\columnwidth}{!}{%
	\begin{tabular}{c|r||ll|l||ll|l}
		\toprule
		\multirow{2}{*}{\bf QId} & \multicolumn{1}{|c||}{\bf \name}		
		& \multicolumn{2}{|c|}{\bf Computation Time} & \multirow{2}{*}{\bf Gain} & \multicolumn{2}{|c|}{\bf Maintenance Time} & \multirow{2}{*}{\bf Gain} \\ 		
		& \multicolumn{1}{|c||}{\bf answer \%} & \bf Adaptive & \bf \ctd & & \bf Recompute & \bf Incremental &  \\
		\midrule
	$Q_{1}$	& $100.00$ & $2.7\times 10^2$  & $6.9\times 10^4$ & $0.99$ & $6.99\times 10^2$ & $6.29\times 10^2$ & $0.100$\\
	
		$Q_{3}$ & $99.86$ & $2.1\times 10^4$ & $1.4\times 10^6$ & $0.98$ & $3.26\times 10^4$ & $2.43 \times 10^4$ & $0.256$ \\
		
		$Q_{6}$	& $100.00$ & $3.9\times 10^3$ & $2.3\times10^6$ & $0.99$ & $1.44\times 10^4$  & $1.11 \times 10^4$ & $0.226$\\ 		
		\bottomrule
	\end{tabular}
	}
	\label{tab:yagooverall}
\end{table}


\begin{table}[t]
	\tabcaption{Performance of adaptive system (in $\mu$seconds) for query computation and maintenance for \gmark query set.}
	\centering
	\resizebox{\columnwidth}{!}{
	\begin{tabular}{c|r||ll|l||ll|l}
		\toprule
		\multirow{2}{*}{\bf QId} & \multicolumn{1}{|c||}{\bf \name}		
		& \multicolumn{2}{|c|}{\bf Computation Time} & \multirow{2}{*}{\bf Gain} & \multicolumn{2}{|c|}{\bf Maintenance Time} & \multirow{2}{*}{\bf Gain} \\ 		
		& \multicolumn{1}{|c||}{\bf answer \%} & \bf Adaptive & \bf \ctd & & \bf Recompute & \bf Incremental &  \\
		\midrule
		$Q_{4}$ & $100.00$ & $5.5\times 10^5$ & $1.0\times 10^7$ & $0.94$ & $6.2\times 10^5$ & $5.6\times 10^5$ & $0.097$\\
		
		$Q_{6}$ & $91.47$ & $4.1\times 10^5$ & $3.9\times 10^6$ & $0.89$ & $4.7\times 10^5$ & $4.4\times 10^5$ & $0.060$\\
		
		$Q_{9}$ & $94.32$ & $1.5\times 10^4$ & $2.6\times 10^5$ & $0.94$ & $2.0\times 10^4$  & $1.8\times 10^4$ & $0.095$\\
		
		$Q_{21}$ & $98.78$ & $6.8 \times 10^4$ & $3.5\times 10^5$ & $0.81$ & $1.4 \times 10^5$ & $1.3\times 10^5$ & $0.091$\\
		
		$Q_{23}$ & $98.88$ & $1.1\times 10^5$ & $8.8\times 10^6$ & $0.98$  & $9.1\times 10^4$ & $7.8\times 10^4$ & $0.14$ \\
		
		$Q_{32}$ & $100.00$  & $1.5\times 10^2$ & $2.4\times 10^5$ & $0.99$&  $1.5\times 10^3$ & $1.3\times 10^3$ & $0.12$ \\
		
		$Q_{35}$ &  $100.00$ & $5.1\times 10^1$ & $8.0\times 10^4$ & $0.99$ & $7.1\times 10^2$ & $6.0\times 10^2$ & $0.16$\\
		
		$Q_{38}$ &  $92.45$ & $2.1\times 10^6$ & $2.1\times 10^7$ & $0.90$ & $2.3\times 10^6$  & $2.2\times 10^6$ & $0.05$ \\
		
		$Q_{46}$ &  $36.92$ &$3.6\times 10^7$ & $4.6\times 10^7$ & $0.21$ & $3.7\times 10^7$  & $3.6\times 10^7$ & $0.005$\\
		
		$Q_{54}$ &  $100.00$  &$4.7\times 10^3$ & $8.5\times 10^5$ & $0.99$  & $1.7\times 10^4$  & $1.3\times 10^4$ & $0.24$ \\
		
		$Q_{90}$ &  $52.27$ & $1.2\times 10^5$ & $1.5\times 10^5$ & $0.20$ & $1.2\times 10^5$ & $1.2\times 10^5$ & $0.008$\\
		\bottomrule
	\end{tabular}
	}
	\label{tab:gmarkoverall}
\end{table}

\subsection{Maintenance Time}
\label{sec:maintenance}

Next we evaluate the maintainability of \name. Notice that in our maintenance
algorithms (Sec.~\ref{subsec:maintain}) we make use of the incremental build-up
of the symbolic probability expressions in our semiring only for the case of
edge addition. When an edge is deleted we simply recompute the probability
expression from scratch. Similarly, updating the probability value for an edge
does not affect the already computed symbolic expression for the probability of
any answer. Thus, we focus on handling edge insertions in the KG.

\subsubsection*{Addition of an Edge}

When quoting maintenance times for our adaptive system we use \name
incrementally on the answers that lie in its domain and do a full
re-computation with \ctd on the rest. This is compared against
re-computation with the adaptive system (\name and \ctd in their
respective domains) for all answers.  We investigated the effect of
the addition of an edge on each individual answer. For every answer,
we randomly selected an edge that affects the answer. The sum total of
the times taken to incrementally update individual answers
(Incremental) for each query is reported in
Table~\ref{tab:yagooverall} and Table~\ref{tab:gmarkoverall}. This is
contrasted against the sum total of the re-computation time
(Recompute) for each individual answer under the specific edge
addition. \comment{We also report the average number of derivations
  added to an answer for every query. These numbers are reported for
  the \name domain of answers in the \emph{adaptive} strategy outlined
  in Sec.~\ref{sec:adaptive}.}

We also report, in the same tables, the percentage of answers for the
entire query that are computed using \name. Since \ctd is not
maintainable, the necessitated re-computation on the answers for which
the adaptive system uses \ctd markedly pulls down the overall gain for
queries with large answers. While we got $10$-$25\%$ gains for queries
that use \name for all their answers, we get more than $5\%$ gains for
queries where at least $90\%$ of the answers are computed by
\name. Thus, the proposed adaptive system inherits the maintainability
of \name.

We also investigated the maintainability of \name. We compared, on the
\name domain of answers, incremental time (IncrTime) versus
re-computation (RecompTime) with \name in Table~\ref{tab:insertyago}
and Table~\ref{tab:insertgmark}.  Notice that addition of an edge can
result in a variable number of derivations being added to an
answer. Since the incremental algorithm adds one derivation at a time
to our symbolic probability expression, its gains over re-computation
go down as the average number of added derivations per answer goes
up. We report a gain of at least $20\%$ for all queries where average
number of derivations added is less than $4$ but greater than
$2$. Note that for answers with just $2$ derivations, re-computation is
almost equivalent to the iterative step. Thus the gains ($10$-$35\%$) are
muted for queries with smaller averages ($< 2$). Queries with
higher average number of derivations have more modest gains
($8$-$11\%$).

\begin{table}[t]
	\tabcaption{Comparison of incremental update time and re-computation time (in $\mu$seconds) for \yago query answers in \name domain.}
		\centering
	{\small
		\begin{tabular}{l|l|ll|l}
			\toprule
			\bf QId  &  \multicolumn{1}{|p{2.5cm}|}{\bf Average number of derivations affected}  & \bf RecompTime  &  \bf IncrTime  & \bf Gain  \\ 
			\midrule
		$Q_{1}$ & $1.33$	 & $6.99 \times 10^2$ & $6.29 \times 10^2$ &	$0.101$	\\
			
		$Q_{3}$ & $1.37$	 & $2.95 \times 10^4$ & $2.11 \times 10^4$ &	$0.284$	\\ 
			
		$Q_{6}$ & $1.55$	 & $1.44 \times 10^4$ & $1.12 \times 10^4$ &	$0.226$	\\		
			\bottomrule
		\end{tabular}
	}
	\label{tab:insertyago}
\end{table}

\begin{table}[t]
	\tabcaption{Comparison of incremental update time and re-computation time (in $\mu$seconds) for \gmark query answers in \name domain.}
	\centering
	{\small
		\begin{tabular}{l|l|ll|l}
			\toprule
			\bf QId  &  \multicolumn{1}{|p{2.5cm}|}{\bf Average number of derivations affected}  & \bf RecompTime  &  \bf IncrTime  & \bf Gain  \\ 	
					\midrule
			
	$Q_{4}$ & $6.815$ & $6.23\times 10^5$ & $5.62\times 10^5$ & $0.097$ \\
		
	$Q_{6}$ & $2.015$	&$9.79\times 10^4$ & $7.00\times 10^4$ & $0.285$ \\
		
	$Q_{9}$ & $1.530$ & $6.98\times 10^3$ & $5.46\times 10^3$ & $0.217$ \\
		
	$Q_{21}$ & $6.848$ & $1.10 \times 10^5$ & $9.70\times 10^4$ & $0.118$ \\
		
	$Q_{23}$ & $1.572$ & $3.80\times 10^4$ & $2.50\times 10^4$ & $0.342$  \\
		
	$Q_{32}$ & $1.000$  &$1.48\times 10^3$ & $1.30\times 10^3$ & $0.122$  \\
		
	$Q_{35}$ & $1.433$ & $7.12\times 10^2$ & $5.95\times 10^2$ & $0.164$ \\
		
	$Q_{38}$ & $1.320$ & $3.29\times 10^5$ & $2.27\times 10^5$ & $0.310$  \\
		
	$Q_{46}$ & $2.633$ &$5.28\times 10^5$ & $3.28\times 10^5$ & $0.379$ \\
		
	$Q_{54}$ & $2.048$ &$1.67\times 10^4$ & $1.27\times 10^4$ & $0.239$  \\
		
	$Q_{90}$ & $4.267$ & $1.01\times 10^4$ & $9.31\times 10^3$ & $0.078$ \\

			\bottomrule						
		\end{tabular}
	}
	\label{tab:insertgmark}
\end{table}

%% file: closing.tex
\section{Conclusions}
\label{sec:concl}

In this paper we have proposed a novel commutative semiring which
enables us to symbolically compute probability of query answers over
probabilistic knowledge graphs. Further, we present a framework \name
that uses the proposed semiring to support query processing and answer
probability maintenance over probabilistic KG.  We have compared the
efficiency of our proposed probability computation technique against
two standard approaches used for probabilistic inference. 
\name outperforms current
systems in a range of queries answer parameters containing almost
$70\%$ of the query answers. We have also shown that an
\emph{adaptive} approach that uses \name in conjunction with a
knowledge compilation based technique for large query answers is a
significantly faster alternative for probabilistic inference.

\comment{

We have proposed a novel approach for probabilistic inference on
Knowledge Graphs. This is backed by a robust theoretical
foundation. Our framework comprises of maintainable symbolic
computations of probabilities of query answers in the semiring of
\emph{flat} integer polynomials. This computation method far
outperforms current systems in a range of queries answer parameters
containing almost $70\%$ of the query answers we processed. We have
made a case for the use of this method in conjunction with a
compilation based technique for large query answers as the optimal way
for probabilistic inference.
}


%% file: appendix.tex
\section{Proof of Theorem 1}\label{appendix1}
The following facts about the function $flat$ are easy to prove and
will be used later.

\begin{lemma} \label{lemma:flatfacts}
  For integer polynomials $f$ and $g$,
  \begin{align*}
    \overline{\overline{f}} &= \overline{f} \\
    \overline{f + g} &= \overline{f} + \overline{g} \\
    \overline{\overline{f}g} &= \overline{fg}
  \end{align*}
\end{lemma}

We now prove our main theorem
\begin{theorem}\label{thm:semiring}
  $(Z_F[p_1, p_2, \ldots ,p_n], \oplus, \otimes, 0,1)$ is a commutative semiring.
\end{theorem}

\begin{proof}
  That $(Z_F[p_1, p_2, \ldots ,p_n], \otimes)$ is a commutative monoid
  with identity $1$ follows directly from the definitions.

  The operator $\oplus$ is also clearly commutative and $0$ is an
  identity for it. To show that $\oplus$ is associative, consider $f,
   g, h \in Z_F[p_1, p_2, \ldots ,p_n]$
  \begin{align*}
    (f \oplus g) \oplus h &= (f + g - \overline{fg}) \oplus h\\
    &= f + g - \overline{fg} + h \\
    &\phantomrel{=}- \overline{fh + gh - h\times flat(fg}) \\
    &= f + g + h - \overline{fg + fh + gh} \\
    &\phantomrel{=} + \overline{h\times \overline{fg}} \\
    &= f + g + h - \overline{fg + fh + gh} + \overline{fgh} \\
    &= (g \oplus h) \oplus f \tag{by symmetry} \\
    &= f \oplus (g \oplus h)
  \end{align*}
  This equality chain is just a repeated application of Lemma
  \ref{lemma:flatfacts}.

Since $0\otimes f = f \otimes 0 = 0$ follows directly by definition,
we are only left to show distributivity of $\otimes$ over $\oplus$. We
prove a quick lemma before we prove distributivity.

\begin{lemma} For all $f \in Z_F[p_1, p_2, \ldots ,p_n]$,
  $$f\otimes f = f$$
\end{lemma}
\begin{proof}
  Given the definition of $Z_F[p_1, p_2, \ldots ,p_n]$ we prove this
  by structural induction.

  By definition of the $\otimes$ operator, $1\otimes 1 = 1$ and $p_i
  \otimes p_i = p_i$ for all $0 \le i \le n$

  It is enough to show that if $g, h \in Z_F[p_1, p_2, \ldots ,p_n]$
  are such that $g \otimes g = g$ and $h \otimes h = h$ then
  $$(g \otimes h) \otimes (g \otimes h) = g \otimes h$$
  and
  $$(g \oplus h) \otimes (g \oplus h) = g \oplus h$$

  Liberally rewriting using Lemma \ref{lemma:flatfacts} and using $g
  \otimes g = \overline{g^2} = g$ and $h \otimes h = \overline{h^2} =
  h$ from our hypothesis, we get
  
  \begin{align*}
    (g \otimes h) \otimes (g \otimes h) &= \overline{\overline{gh}\times \overline{gh}} \\
    &= \overline{g^2 h^2} \\
    &= \overline{\overline{g^2} \times \overline{h^2}} \\
    &= \overline{gh} = g \otimes h
  \end{align*}
  and
  \begin{align*}
    (g \oplus h) \otimes (g \oplus h) &= (g + h - \overline{gh})\otimes (g + h - \overline{gh}) \\
    &= \overline{g^2 + h^2 + 2gh - 2g\overline{gh} - 2h\overline{gh} + (\overline{gh})^2} \\
    &= \overline{g^2} + \overline{h^2} + 2\overline{gh} - 2\overline{g\overline{gh}} - 2\overline{h\overline{gh}} + \overline{\overline{gh}^2} \\
    &= g + h + 2\overline{gh} - 2\overline{g^2h} - 2\overline{gh^2} + \overline{\overline{g^2h^2}} \\
    &= g + h + 2\overline{gh} - 2\overline{\overline{g^2}h} - 2\overline{g\overline{h^2}} + \overline{\overline{g^2} \times \overline{h^2}} \\
    &= g + h + 2\overline{gh} - 2\overline{gh} - 2\overline{gh} + \overline{gh} \\
    &= g + h - \overline{gh} \\
    &= g \oplus h
  \end{align*}
\end{proof}

To wrap up our proof of theorem~\ref{thm:semiring}, we now prove left
distributivity of $\otimes$ over $\oplus$. For $f, g, h \in Z_F[p_1,
  p_2, \ldots ,p_n]$
\begin{align*}
  f \otimes (g \oplus h) &= \overline{f(g + h - \overline{gh})} \\
  &= \overline{fg} + \overline{fh} - \overline{f\overline{gh}} \\
  &= \overline{fg} + \overline{fh} - \overline{\overline{f^2} \times \overline{gh}} \\
  &= \overline{fg} + \overline{fh} - \overline{fg \times fh} \\
  &= \overline{fg} + \overline{fh} - \overline{\overline{fg} \times \overline{fh}} \\
  &= \overline{fg} \oplus \overline{fh} = (f \otimes g) \oplus (f \otimes h) 
\end{align*} \qed
  
\end{proof}